%% file: sqr2.tex
\documentclass[12pt]{article}
\usepackage{amsmath}
\usepackage{graphicx}
\usepackage{enumerate}
\usepackage{natbib}
\usepackage{amsmath}
\usepackage{url} 

\newcommand{\blind}{0}

\input{myfonts}

\usepackage{float}
\usepackage{caption}
\usepackage{subcaption}
\usepackage{booktabs}

\def\spacingset#1{\renewcommand{\baselinestretch}%
{#1}\small\normalsize} \spacingset{1}

 \DeclareMathSymbol{,}{\mathpunct}{operators}{"2C}

\begin{document}

\if0\blind
{
  \title{\bf  Spline Quantile Regression with Cubic and Linear Smoothing Splines}
  \author{Ta-Hsin Li\footnote{Formerly affiliated with IBM T.\ J.\ Watson Research Center, Yorktown Heights, NY 10598, USA. Email: {\sc thl024@outlook.com}}
  }
\date{March 23, 2026}
  \maketitle
} \fi

\if1\blind
{
  \title{\bf    Spline Quantile Regression with Cubic and Linear Smoothing Splines} 
  \maketitle
} \fi

\begin{abstract}
Spline quantile regression (SQR) is a method introduced recently by Li and Megiddo (2026)  
for linear quantile regression where the regression coefficients are treated as smooth functions of the quantile level. 
With the coefficients represented by cubic splines with fixed knots on a given set of quantiles, 
the SQR method produces an estimate for the functional coefficients by solving a penalized quantile 
regression problem. The $\ell_1$-norm of the second derivatives of the coefficients is employed 
as the penalty for regulating the roughness of the functional coefficients. This extends 
the SQR method by introducing additional pairings of the functional representation 
for the regression coefficients and the penalty for their roughness. The resulting cubic and  linear SQR solutions 
are shown to be smoothing splines which are optimal in a 
functional space  larger than the respective spline space with fixed knots.  
It is shown that the cubic SQR can be reformulated and solved as a quadratic program and the linear SQR as a linear program. A simulation study demonstrates that the SQR solutions  not only offer a concise functional representation of the regression coefficients  with distinct smoothness characteristics, but also provide a capability of producing more accurate estimates of the regression coefficients when the underlying functions are suitably smooth. Application of the SQR solutions is demonstrated by real-data examples, including a Granger causality analysis of stock market indices.


\bigskip\bigskip
\noindent
{\it Acknowledgment}: The author would like to thank Dr.\ R.\ Koenker for advice on  
the use of FORTRAN code {\tt rqfnb}.
\vfill

\end{abstract}

\newpage
\spacingset{1.9} 

\section{Introduction}

Quantile regression  (QR) is a statistical method that complements the conventional least-squares regression with the ability to uncover the effect of explanatory variables on a dependent variable across its conditional distribution (Koenker and Bassett 1978; Koenker 2005). 
Let $\{ (\bx_t, y_t): t=1,\dots,n\}$ be a set of observations with $y_t$ denoting 
the dependent variable in $\bbR$ and $\bx_t$ denoting the corresponding vector of 
explanatory variables in $\bbR^p$ for some fixed $p \ge 1$. Let 
$F(y \mid \bx)  := \Pr\{ y_t \le y \mid \bx_t = \bx\}$ be the conditional distribution function of $y_t$ given $\bx_t = \bx$. Under the assumption that the conditional quantile $F^{-1}(\tau \mid \bx)$ at a given quantile level $\tau \in (0,1)$  is a linear function of $\bx$, i.e., $F^{-1}(\tau \mid \bx) = \bx^T \bmbeta_0$ for some $\bmbeta_0 \in \bbR^p$, 
the QR method produces an estimate of  $\bmbeta_0$ by solving
\eqn
\hat{\bmbeta}
:=  \operatorname*{argmin}_{\bmbeta \in \bbR^p} \sum_{t=1}^n
\rho_{\tau}(y_t - \bx_t^T \bmbeta),
\label{QR}
\eqqn
where $\rho_\tau(y) := y (\tau - I(y < 0))$ and $I(\cdot)$ the indicator 
function.  In the special case of $\bx_t = 1$ for all $t$, the QR solution in (\ref{QR})
reduces to the $\tau$th sample quantile of $\{ y_t: t=1,\dots,n\}$.

The QR method can be extended to explore the effect of $\bx_t$ on $y_t$ at multiple quantile levels 
$\tau_\ell$ $(\ell=1,\dots,L)$, assuming
$F^{-1}(\tau_\ell \mid \bx) = \bx^T \bmbeta_{\ell}$ for some $\bmbeta_\ell \in \bbR^p$ $(\ell=1,\dots,L)$.
In this case, the unknown parameters $\bmbeta_{\ell}$ ($\ell=1,\dots,L$) are estimated 
by solving the QR problem (\ref{QR}) for each quantile level  independently, i.e.,
\eqn
\hat{\bmbeta}_{\ell}
:=  \operatorname*{argmin}_{\bmbeta \in \bbR^p} \sum_{t=1}^n
\rho_{\tau_\ell}(y_t - \bx_t^T \bmbeta) \quad (\ell=1,\dots,L).
\label{QR2}
\eqqn
Special algorithms exist that allow faster computation of the QR solutions at multiple quantiles 
(Koenker and D'Orey 1987; Chernozhukov et al.\ 2022).

In this article, we are interested in the situation where the conditional quantile $F^{-1}(\tau \mid \bx)$
is a linear function of $\bx$ for all $\tau$ in a closed interval $[a,b] \subset (0,1)$, i.e.,
\eqn
F^{-1}(\tau \mid \bx) = \bx^T \bmbeta_0(\tau) \quad \forall \tau \in [a,b],
\label{F}
\eqqn
where $\bmbeta_0(\cdot)$ is a continuous function. The objective
is to estimate $\bmbeta_0(\cdot)$ as a function on the basis of the data record $\{ (\bx_t,y_t): t=1,\dots,n\}$. The assumed continuity of $\bmbeta_0(\cdot)$ implies that $F^{-1}(\cdot \mid \bx)$ 
is a continuous function for each fixed $\bx$. An example is when $y_t$ has a conditional probability density function $\dot{F}(\cdot \mid \bx)$ which is everywhere positive. This functional view 
of quantile regression has two advantages over the conventional view of parameter estimation at a fixed quantile:
First, it provides a complete description  of the effect of the explanatory variables across all quantiles, as is 
necessary when the conditional quantile function is treated as a decomposition of the conditional mean in the sense that
\eq
\E\{y_t \mid \bx_t \} = \lim_{
\substack{a \rightarrow 0^+,\ b \rightarrow 1^-}}
 \int_a^b F^{-1}(\tau \mid \bx_t) \, d\tau.
\eqq
Secondly,  it facilitates more accurate assessment of the effect 
at a given quantile using the information from neighboring quantiles. 

A simple functional estimate of $\bmbeta_0(\cdot)$  can be obtained by first computing
the QR estimates $\hat{\bmbeta}(\tau_\ell) := \vect[\hat{\beta}_j(\tau_\ell)]_{j=1}^p := \hat{\bmbeta}_{\ell}$  $(\ell=1,\dots,L)$  from (\ref{QR2})
for a suitable set of quantile levels $\tau_1 := a < \tau_2 < \dots < \tau_{L-1} < \tau_L :=b$
and then interpolating the point estimates $\hat{\beta}_j(\tau_\ell)$ $(\ell=1,\dots,L\}$ for each element $j \in \{1,\dots,p\}$ by a suitable interpolator (Stoer and Bulirsch 2002; Neocleous and Portnoy 2008). With noisy data, 
using a smoother  (Hastie and Tibshirani 1990) instead of an interpolator tends to produce 
more accurate estimates,  an example being the kernel smoother employed by Koenker (2005, p.\ 158).

In this article, we consider an alternative method, called spline quantile regression (SQR), 
for estimating the function $\bmbeta_0(\cdot)$. Instead of post-smoothing with an ad hoc procedure 
as described above, the SQR method combines the quantile regression problem with the 
smoothing problem under a single framework of penalized quantile regression.  A suitable 
functional space is employed to represent the regression coefficients and a suitable penalty is 
employed to regulate their roughness. This idea  was originally explored in Li and Megiddo (2026), where
the functional space consists of cubic splines with fixed knots at a given set of quantiles  
and the penalty takes the form of the $\ell_1$-norm of the second derivatives of the functional coefficients.
In this article, we extend the SQR method by offering additional choices for the pairing of the functional space and
the penalty that lead to functional estimates with different smoothness characteristics. We show that the resulting cubic SQR problem  can be reformulated as a quadratic program (QP), whereas 
 the resulting linear SQR problem can be reformulated as a linear program (LP) similar to the SQR problem in Li and Megiddo (2026). Furthermore, we show that these new solutions 
have the additional property of being optimal in a functional space larger than the respective space of 
cubic or linear splines with fixed knots at the given set of quantiles. This optimality property is analogous to  the 
theory of smoothing splines for nonparametric least-squares regression (Wahba. 1975).

 The SQR problem  is different from the problem of quantile smoothing spline  considered by Koenker et al.\ (1994). While both employ spline functions in quantile regression, 
the latter produces a function of the explanatory variable as the regressor in a  nonlinear quantile regression problem at a fixed quantile, whereas the former produces a function of the quantile level as the coefficient of the explanatory variable under a linear quantile regression setting.
The SQR problem also differs from the varying-coefficient problem considered by
Andriyana et al.\ (2014) where conventional quantile regression is performed 
at a fixed quantile and the regression coefficients are treated as spline functions of a covariate.
Park and He (2017) employ a model similar to ours but the estimation is performed without regulation of any kind. Yoshida (2021)  introduces a group lasso type penalty on the spline parameters  for the purpose of variable selection, rather than a roughness penalty on the regression coefficients for smoothing.
See Li and Megiddo (2026) for comments on additional spline-based methods.

This article focuses on the development of the SQR method and its computation, leaving the statistical theory for future work. The remainder of this article is organized as follows. In Section~2, we describe the  cubic and linear 
SQR solutions and discuss their optimality in the respective functional spaces.
In Section~3, we discuss computational issues of the cubic and linear SQR solutions, including
their reformation as quadratic and linear programs, respectively,  the automatic selection 
of smoothing parameter, and the construction of bootstrap confidence band.
In Sections~4 and 5, we report the results of a simulation study and application examples with real data, 
including a Granger causality analysis of stock market indices. Concluding remarks are given in Section~6. A description of the R functions that implement the SQR method is given in Appendix. The supplementary material contains additional real-data examples including 
a quantile spectral analysis of sunspot numbers and a quantile regression analysis of US birth data.

\section{Spline Quantile Regression}

The SQR method in general produces an estimate of $\bmbeta_0(\cdot)$ in (\ref{F}) by solving 
a penalized linear quantile regression problem with functional coefficients:
\eqn
\hat{\bmbeta}(\cdot) := \operatorname*{argmin}_{\bmbeta(\cdot) \in \cF} \bigg\{
n^{-1}
\sum_{\ell=1}^L \sum_{t=1}^n
\rho_{\tau_\ell}(y_t - \bx_t^T \bmbeta(\tau_\ell)) +  c \, R(\bmbeta(\cdot)) \bigg\}.
\label{sqr}
\eqqn
In this expression, $\cF$ is a suitable functional space spanned by spline functions,
$R(\bmbeta(\cdot))$ is a suitable measure of roughness for $\bmbeta(\cdot)$ 
as a function in $\cF$, and $c > 0$ is a suitable smoothing parameter that controls the amount 
of penalty. If $\cF$ is capable of  interpolation, then, in the limiting case as $c \lzero$ the SQR solution  
will reduce to a function that interpolates the QR estimates 
given by (\ref{QR2}), i.e.,  $\hat{\bmbeta}(\tau_\ell) = \hat{\bmbeta}_\ell$ 
for all $\ell =1,\dots,L$. Generally, with $c > 0$, the SQR solution trades off the fidelity to the QR estimates with 
the smoothness and serves as a smoother in effect.

The SQR method was originally investigated in Li and Megiddo (2026) with a specific choice 
of $\cF$ being the space spanned by cubic splines with knots $\{ \tau_\ell\}$  and $R(\bmbeta(\cdot))$ 
being the integral of the $\ell_1$-norm of the second derivative of $\bmbeta(\cdot)$  on $[a,b]$. 
This choice results in a linear program that can be solved efficiently  in a similar way to the ordinary QR problem.
While being a valid estimate, the optimality of this solution is not entirely aligned with the existing theory of smoothing splines when considered in the context of a larger functional space.
 In the following, we expand the scope of the original investigation by considering two additional choices for the pairing of $\cF$ and $R(\bmbeta(\cdot))$.  These choices lead to a cubic spline solution and a linear spline solution, respectively.
We show that the optimality of these solutions extends beyond 
the space of cubic or linear splines with fixed knots  $\{ \tau_\ell\}$.

\subsection{Cubic Spline Quantile Regression}
\label{sec:sqr3}

First, consider the case where $\cF$ in (\ref{sqr})  is the  space spanned by cubic splines with knots  $\{ \tau_\ell\}$ and $R(\bmbeta(\cdot))$ in (\ref{sqr})  takes the form
\eqn
R(\bmbeta(\cdot)) :=  \int_a^b \| \ddot{\bmbeta}(\tau) \|_2^2 \, d\tau,
\label{R2}
\eqqn
where $\ddot{\bmbeta}(\cdot)$ 
denotes the second derivative of $\bmbeta(\cdot)$. The resulting SQR solution will be referred 
to as the cubic SQR  solution or estimator.

\medskip
The cubic SQR solution is optimal not only 
in the space of cubic spline functions with knots $\{ \tau_\ell\}$ but beyond.  In fact, it is the optimal 
solution in a functional space which comprises all continuously differentiable functions 
with square-integrable second derivative, i.e.,
\eq
\cF_2[a,b] := \{ \bmbeta(\cdot) \in C[a,b]: \dot{\bmbeta}(\cdot) \in C[a,b], 
\ddot{\bmbeta}(\cdot) \in L_2[a,b]\}.
\eqq
The following lemma, known as the minimum property of cubic splines,  describes the minimizer of the penalty in (\ref{R2}) over the space $\cF_2[a,b]$ when constrained by the values at $\{ \tau_\ell\}$.

\begin{lem}[de Boor 1963]
{\rm
For any fixed $\bmbeta_\ell \in  \bbR^p$ $(\ell =1,\dots,L)$, there exists a unique cubic spline function 
$\tilde{\bmbeta}(\cdot)$ 
with knots $\tau_\ell$ $( \ell=1,\dots,L)$ that minimizes $R(\bmbeta(\cdot))$ in (\ref{R2}) among 
all functions in $\cF_2[a,b]$ 
satisfying $\bmbeta(\tau_\ell) = \bmbeta_\ell$ $(\ell=1,\dots,L)$.
}
\end{lem}

\noindent
This result implies that minimizing the loss function in (\ref{sqr})  over 
$\bmbeta(\cdot)$ in $\cF_2[a,b]$ is equivalent to minimizing $n^{-1}
\sum_{\ell=1}^L \sum_{t=1}^n
\rho_{\tau_\ell}(y_t - \bx_t^T \tilde{\bmbeta}(\tau_\ell)) +  c \, R(\tilde{\bmbeta}(\cdot))$
with respect to $\tilde{\bmbeta}(\cdot)$ given by Lemma~1.
Therefore, it suffices to solve (\ref{sqr}) by restricting $\cF$ to be
the space of cubic spline functions  with knots $\{ \tau_\ell \}$, which leads to the cubic SQR solution.
The optimality of this solution beyond what is implied by its formulation can be appreciated from two perspectives: first, it is not confined to cubic splines, and secondly, it does not restrict the placement of knots at $\{ \tau_\ell \}$.

\subsection{Linear Spline Quantile Regression}
\label{sec:sqr1}

Next, consider the case where $\cF$ in (\ref{sqr}) is the space spanned by linear splines with knots $\{ \tau_\ell\}$ and  $R(\bmbeta(\cdot))$ in (\ref{sqr})  takes the form
\eqn
R(\bmbeta(\cdot)) := \int_a^b \| \ddot{\bmbeta}(\tau) \|_1 \, d\tau.
\label{R1}
\eqqn
With this choice, any function $\bmbeta(\cdot)$ in $\cF$ is a continuous piecewise linear function
with a constant derivative $\dot{\bmbeta}(\tau_\ell)$ in the interval $[\tau_\ell,\tau_{\ell+1})$ $(\ell=1,\dots,L-1)$. 
Moreover, if we employ  the Dirac delta $\del(\tau)$ to represent the derivative of the step function $I(\tau > 0)$ and define $ \dot{\bmbeta}(\tau_{L}) := \0$, then
the second derivative of $\bmbeta(\cdot)$ can be written as
\eq
\ddot{\bmbeta}(\tau) = \sum_{\ell=1}^{L-1} (\dot{\bmbeta}(\tau_{\ell+1})-\dot{\bmbeta}(\tau_\ell)) 
\del(\tau - \tau_{\ell+1}).
\eqq
Therefore, the penalty in (\ref{R1}) is equivalent to
\eqn
R(\bmbeta(\cdot)) = \sum_{\ell=1}^{L-1}  \|  \dot{\bmbeta}(\tau_{\ell+1}) -  \dot{\bmbeta}(\tau_{\ell}) \|_1.
\label{TV2}
\eqqn
This quantity  is nothing but the total amount of change in the slopes of piecewise linear functions.  
The resulting SQR solution  will be referred to as the linear SQR solution or estimator.
This solution  differs from  that of Li and Megiddo (2026) as the latter 
employs cubic rather than linear splines, for which the $\ell_1$-norm penalty in (\ref{R1}) does not reduce to (\ref{TV2}).

The linear SQR solution is  optimal in a larger space than that 
spanned by linear splines with knots $\{ \tau_\ell \}$. 
The formal mathematical theory behind this claim is rather complicated because the space naturally
associated with the $\ell_1$-norm penalty (\ref{R1}), 
\eq
\cF_1[a,b] := \{ \bmbeta(\cdot) \in C[a,b]: \dot{\bmbeta}(\cdot) \in C[a,b], \ddot{\bmbeta}(\cdot) \in L_1[a,b]\},
\eqq
does not contain linear 
splines in the strict sense. This difficulty was pointed out by Fisher and Jerome (1975)  (cf.\  Brezis 2019). Their remedy was to extend the space  by admitting functions whose second derivatives are measures. 
It results in a functional space which we denote by  $\bar{\cF}_1[a,b]$.
This extension also requires that the $\ell_1$-norm of the second derivative be replaced by 
the total variation of the first derivative $\dot{\bmbeta}(\cdot) := [\dot{\beta}_1(\cdot),\cdots,\dot{\beta}_p(\cdot)]^T$. In other words, let
\eqn
R(\bmbeta(\cdot))  := \sum_{j=1}^p V_a^b(\dot{\beta}_j(\cdot)).
\label{TV}
\eqqn
This penalty reduces to  (\ref{R1})  for functions with an integrable second derivative.

The  aforementioned extension method was employed by Koenker et al.\ (1994)  to justify a linear spline solution
for nonparametric quantile regression at a fixed quantile.  Their argument, which is based on the result 
of Pinkus (1988), leads to the following lemma.

\begin{lem}[Koenker et al.\ 1994; Pinkus 1988]
{\rm 
For any fixed $\bmbeta_\ell \in \bbR^p$ $(\ell =1,\dots,L)$, there exists a linear spline function $\tilde{\bmbeta}(\cdot)$ with knots $\tau_\ell$ $(\ell=1,\dots,L)$ that minimizes $R(\bmbeta(\cdot))$ in (\ref{TV}) 
over $\bar{\cF}_1[a,b]$ under the constraint $\bmbeta(\tau_\ell) = \bmbeta_\ell$ $(\ell=1,\dots,L)$.
}
\end{lem}

\noindent
According to Lemma 2, the minimizer of the loss function in (\ref{sqr}) over $\bar{\cF}_1[a,b]$ 
with $R(\bmbeta(\cdot))$ given by (\ref{TV}) can be found in the space of linear spline functions 
with knots  $\{ \tau_\ell\}$.  If $\bmbeta(\cdot)$ is such a function, then $\dot{\beta}_j(\cdot)$ is a
piecewise constant function with possible jumps at $\tau_\ell$ $(\ell=2,\dots,L-1)$. This implies that $V_a^b(\dot{\beta}_j(\cdot)) = \sum_{\ell=1}^{L-1}  |  \dot{\beta}_j(\tau_{\ell+1}) -  \dot{\beta}_j(\tau_{\ell}) |$.
Therefore, it suffices to solve (\ref{sqr}) with $R(\bmbeta(\cdot))$ given by (\ref{TV2}) and $\cF$ 
restricted to the space spanned by linear splines with knots $\{\tau_\ell \}$, resulting in the linear SQR solution. 
Similarly to the cubic SQR solution, the optimality of the linear SQR solution extends beyond what is implied by its formulation: it is not limited to linear spline
functions and does not restrict the placement of knots to $\{ \tau_\ell \}$.

\section{Computation}

The linear SQR solution can be solved as a linear program 
like the ordinary QR problem and the SQR solution in Li and Megiddo (2026).. 
The cubic SQR problem need to be solved as a quadratic program.
This section contains the computational details. It also presents two data-driven criteria for selecting 
the smoothing parameter $c$ for the linear and cubic SQR problems and a bootstrap method for constructing confidence bands.

\subsection{Cubic SQR by Quadratic Programming}

Let $\{ \phi_k(\cdot): k=1,\dots,K\}$ denote a set of cubic spline basis functions defined on the interval $[a,b]$ 
with knots $\{ \tau_\ell \}$  $(K :=L+2)$ . 
Then, for any cubic spline function $\bmbeta(\cdot) := [\beta_1(\cdot),\dots,\beta_p(\cdot)]^T \in \cF$, we can write
\eq
\beta_j(\cdot) = \sum_{k=1}^K \phi_k(\cdot) \theta_{jk} =  \bmphi^T(\cdot) \bmth_j \quad (j=1,\dots,p),
\eqq
where
\begin{gather*}
\bmphi(\cdot)  :=  [\phi_1(\cdot),\dots,\phi_K(\cdot)]^T, \quad \bmth_j \ :=    [\theta_{j1},\dots,\theta_{jK}]^T \in \bbR^{K}.
\end{gather*}
By combining these expressions with the notation
\begin{gather*}
\bPhi(\cdot) :=  \bI_p \otimes \bmphi^T(\cdot), \quad \bmth  :=  [\bmth_{1}^T,\dots,\bmth_{p}^T]^T \in \bbR^{pK},
\end{gather*}
we can write 
\eq
\bmbeta(\cdot) = \bPhi(\cdot) \, \bmth.
\eqq
The penalty in (\ref{R2}) can be written as
\eqn
R(\bmbeta(\cdot)) = 
\bmth^T \bigg\{
\int_a^b  \ddot{\bPhi}^T(\tau) \,  \ddot{\bPhi}(\tau)\, d\tau \bigg\} \, \bmth.
\label{R2:sqr3}
\eqqn
Let
\eqn
\bOm := 2 n c \int_a^b \ddot{\bPhi}^T(\tau) \,  \ddot{\bPhi}(\tau)\, d\tau.
\label{om}
\eqqn
Then, the SQR problem (\ref{sqr}) with $R(\bmbeta(\cdot))$ given by (\ref{R2}) can be restated as
\eqn
\hat{\bmth} :=
\operatorname*{argmin}_{\bmth \in \bbR^{pK}} \bigg\{
 \sum_{\ell=1}^L \sum_{t=1}^n
\rho_{\tau_\ell}(y_t - \bx_t^T \bPhi(\tau_\ell) \bmth) +  (1/2)  \bmth^T  \bOm \, \bmth \bigg\}
\label{sqr3}
\eqqn
and
\eqn
\hat{\bmbeta}(\tau) := \bPhi(\tau)\,\hat{\bmth} \quad \forall \tau \in [a,b].
\label{betah}
\eqqn
In other words, the  cubic SQR solution can be obtained by searching for the vector $\hat{\bmth} \in \bbR^{pK}$ according to (\ref{sqr3}) and then transforming it  to $\hat{\bmbeta}(\cdot)$ using the basis functions according to (\ref{betah}).

Let $\bX  :=  [  \bx_1,\dots,\bx_n]^T$ be the $n$-by-$p$ design matrix. Let 
$\1$ denote a vector of 1's
and $\0$ a vector or matrix of 0's (when necessary a subscript will be introduced to show the dimension). 
With this notation and a standard trick of linearization in quantile regression, 
the SQR problem in (\ref{sqr3}) can be 
reformulated as a convex quadratic program (QP):
\eqn
\lefteqn{
\{\hat{\bmth},\hat{\bu}_1,\dots,
\hat{\bu}_L,  \hat{\bv}_1, \dots, \hat{\bv}_L \}  :=} \notag \\
&  &  
\operatorname*{argmin}_{
\bmth, \bu_1,\dots,\bu_L,\bv_1,\dots,\bv_L } \
\sum_{\ell=1}^L \{ \tau_\ell \1^T \bu_\ell + (1-\tau_\ell) \1^T \bv_\ell\}
+ (1/2) \bmth^T \bOm \, \bmth
   \notag \\
&& \quad
{\rm s.t.} \quad
\left\{
\begin{array}{l} 
\bX \, \bPhi(\tau_\ell) \, \bmth+ \bu_\ell - \bv_\ell = \by  \quad (\ell=1,\dots,L),  \\
\bu_\ell \ge \0, \ \bv_\ell \ge \0 \quad (\ell=1,\dots,L). \\
\end{array} \right.
\label{lp:sqr3}
\eqqn
In this formulation, the additional  decision variables $\bu_\ell \in \bbR_+^n$ and $\bv_\ell \in \bbR_+^n$ 
are introduced  just for the purpose of linearizing the loss function of quantile regression in  (\ref{sqr3}).

To put the QP problem (\ref{lp:sqr3}) into a more compact form, let 
\eq
\bu := [\bu_1^T,\dots,\bu_L^T]^T, \quad \bv := [\bv_1^T,\dots,\bv_L^T]^T,
\eqq
\eqn
\bc  := [\tau_1 \1_n^T, \dots, \tau_L \1_n^T]^T, \quad
\bD := [\bPhi^T(\tau_1) \bX^T,\dots,\bPhi^T(\tau_L) \bX^T]^T,
\label{D}
\eqqn
and
\eqn
 \bb := \1_L \otimes \by.
 \label{b}
\eqqn
Then, the equality constraints in (\ref{lp:sqr3}) become
$\bD \, \bmth + \bu - \bv  = \bb$ and the linear term of the loss function can be written as
$\sum_{\ell=1}^L \{ \tau_\ell \1^T \bu_\ell + (1-\tau_\ell) \1^T \bv_\ell \} 
= \bc^T \bu + (\1- \bc)^T \bv$. Therefore, the QP problem  (\ref{lp:sqr3}) takes the canonical form
\eqn
\min\{ \bc^T \bu + (\1- \bc)^T \bv +  (1/2) \bmth^T \bOm \, \bmth \mid  \bD\, \bmth  + \bu - \bv = \bb; \bmth \in \bbR^{pK}; \bu, \bv \in \bbR_+^{nL}\}.
\label{primal:sqr3}
\eqqn
This problem can be solved numerically by any general-purpose QP solvers that allow a semi-definite matrix in the quadratic term and accommodate linear equality and inequality constraints.

Associated with the primal QP in (\ref{primal:sqr3}) is the dual QP given by the following theorem.
The dual QP can serve as a more convenient form of input to  primal-dual algorithms that produce both primal and dual solutions.

\begin{thm}
{\rm
Let $\ba := \bD^T (\1 -\bc) \in \bbR^{pK}$, where $\bc$ and $\bD$ are given by (\ref{D}).
Let $\bb$ be defined by (\ref{b}). Then, the dual of the primal QP in (\ref{primal:sqr3})  can be formulated as
\eqn
\min\{ -\bb^T \bmzeta + (1/2)  \bmeta^T \bOm \, \bmeta 
\mid  \bD^T \bmzeta -  \bOm  \, \bmeta = \ba;  \bmeta \in \bbR^{pK};\bmzeta \in [0,1]^{nL} \},
\label{dual:sqr3}
\eqqn
where  $\bmlam  := \bmzeta - (\1-\bc)$ is the Lagrange multiplier for the equality constraints in (\ref{primal:sqr3}). 
}
\end{thm}

\noindent
{\sc Proof}. To put the primal QP into the standard  form,  
let $\bmth$ be reparameterized as $\bmgam - \bmdel$ with $\bmgam, \bmdel \in \bbR_+^{pK}$.
Define $\bmxi :=  [\bmgam^T,\bmdel^T,\bu^T,\bv^T]^T$, $\bd :=  [\0_{pK}^T,\0_{pK}^T,\bc^T,(\1 - \bc)^T]^T$, 
$\bA   :=   [ \bD, -\bD, \bI_{nL},-\bI_{nL}]$, and $\bQ := \diag\{(\bI_2 - \tilde{\bI}_2) \otimes \bOm,\0_{2nL \times 2nL} \}$ with $\tilde{\bI}_2$ denoting the 2-by-2 anti-diagonal matrix.
Then, the primal QP in (\ref{primal:sqr3}) can be rewritten as
\eqn
\min\{ \bd^T \bmxi + (1/2)  \bmxi^T \bQ \, \bmxi \mid \bA \, \bmxi = \bb; \bmxi \in \bbR_+^{2pK+2nL} \}.
\label{primal3}
\eqqn
The corresponding dual QP (Dorn 1960; Friedlander and Orban 2012) takes the form
\eqn
\min\{ -\bb^T \bmlam + (1/2)  \bz^T \bQ  \,\bz  \mid    \bA^T \bmlam  - \bQ \, \bz \le \bd;  
\bmlam \in \bbR^{nL}; \bz \in \bbR^{2pK+2nL} \},
\label{dual0}
\eqqn
where  $\bmlam$ is the Lagrange multiplier 
for the equality constraints in (\ref{primal3}). Let $\bz = [\bz_1^T,\bz_2^T,\bz_3^T,\bz_4^T]^T$ 
be partitioned in the same way as $\bmxi = [\bmgam^T,\bmdel^T,\bu^T,\bv^T]^T$.
 Then, with the new variable $\bmeta := \bz_1 - \bz_2 \in \bbR^{pK}$, we can write
\eq
\bQ  \, \bz =
\left[
\begin{array}{r}
\bOm  \, \bmeta  \\
-\bOm \, \bmeta \\
 \0_{nL} \\
\0_{nL}  
\end{array}  
\right].
\eqq
Moreover, by partitioning $\bmlam$ according to the structure of $\bb$
such that $\bmlam := [\bmlam_1^T,\dots,\bmlam_L^T]^T$.
the inequalities $\bA^T \bmlam - \bQ \, \bz \le \bd$ in (\ref{dual0}) can be written more elaborately as
\begin{gather*}
\sum_{\ell=1}^L \bPhi^T(\tau_\ell) \, \bX^T \bmlam_\ell  - \bOm \,  \bmeta \le   \0, \quad
- \sum_{\ell=1}^L  \bPhi^T(\tau_\ell) \, \bX^T \bmlam_\ell + \bOm \,  \bmeta  \le    \0, \\
\bmlam_\ell \le  \tau_\ell \1,   \quad  - \bmlam_\ell \le  (1-\tau_\ell) \1 \quad (\ell=1,\dots,L).
\end{gather*}
These inequalities are equivalent to
\eq
\sum_{\ell=1}^L \bPhi^T(\tau_\ell) \, \bX^T \bmlam_\ell - \bOm \, \bmeta = \0, \quad 
\bmlam_\ell \in [\tau_\ell-1, \tau_\ell]^n \quad (\ell =1,\dots,L).
\eqq
By a change of variables, $\bmzeta   := [\bmzeta_1^T,\dots,\bmzeta_L^T]^T := 
\bmlam +  (\1-\bc)$, i.e., $\bmzeta_\ell := \bmlam_\ell + (1-\tau_\ell) \1$ $(\ell=1,\dots,L)$,
we can write
\begin{gather*}
\bb^T \bmlam  = \sum_{\ell=1}^L \by^T \bmlam_\ell =  \sum_{\ell=1}^L \by^T \bmzeta_\ell  
-\sum_{\ell=1}^L (1-\tau_\ell) \by^T \1 = \bb^T \bmzeta + \text{constant}, \\
\sum_{\ell=1}^L \bPhi^T(\tau_\ell) \, \bX^T \bmlam_\ell  =  \sum_{\ell=1}^L  \bPhi^T(\tau_\ell) \, \bX^T  \bmzeta_\ell  
 - \sum_{\ell=1}^L (1-\tau_\ell) \, \bPhi^T(\tau_\ell) \, \bX^T \1
 =  \bD^T \bmzeta - \ba, \\
\bmlam_\ell \in [\tau_\ell-1,\tau_\ell]^n  \leftrightarrow  \bmzeta_\ell \in [0,1]^n,  \quad (\ell=1,\dots,L).
\end{gather*}
Combining these expressions 
with the fact that $\bz^T \bQ \, \bz = \bmeta^T \bOm \, \bmeta$ proves the assertion. \qed

In the R implementation (see Appendix), we employ  the  {\tt solve\underline{ }piqp} function 
in the R package `piqp' (Narasimhan et al.\ 2023; Schwan et al.\ 2023) to solve the cubic SQR problem. 
This function is a general-purpose solver that uses a proximal  interior-point method to
produce both primal and dual solutions for convex QP problems under linear equality and inequality constraints with  bounded or unbounded decision variables. Experience shows that applying this function to 
the dual  QP in (\ref{dual:sqr3})  solves the cubic SQR problem faster than 
applying it to the primal QP in (\ref{primal:sqr3}).   Another general-purpose solver employed in our R implementation is the  {\tt solve\underline{ }osqp}  function from the `osqp' package (Stellato et al.\ 2020; Stellato et al.\ 2024). As a first-order method, this function is able to accommodate  large matrices that cannot be handled by  {\tt solve\underline{ }piqp}. However, its convergence tends to be more difficult to achieve  especially when the smoothing parameter in the cubic SQR problem is relatively large.

For convenience, the R implementation substitutes the penalty in (\ref{R2:sqr3}) by a discrete approximation
$\bmth^T \{ \sum_{\ell=1}^L w_\ell \, \ddot{\bPhi}^T(\tau_\ell) \ddot{\bPhi}(\tau_\ell) \} \bmth$.
With this choice, the primal and dual formulations discussed in Section~\ref{sec:sqr3} remain
valid provided that $\bOm$ in (\ref{om}) is redefined as
\eq
\bOm := 2 \sum_{\ell=1}^L c_\ell \, \ddot{\bPhi}^T(\tau_\ell) \, \ddot{\bPhi}(\tau_\ell) 
\eqq
with  $c_\ell := n c w_\ell$. The  weight sequence $\{ w_\ell \}$ is introduced to enable user-specified adjustment to  the local smoothness requirement  at different quantiles.

\subsection{Linear SQR by Linear Programming}

Let $\{ \phi_k(\tau): k=1,\dots,K\}$ denote a set of linear spline basis functions defined on the interval $[a,b]$
with knots $\{ \tau_\ell\}$ $(K:= L$). With the matrix notation introduced in Section~\ref{sec:sqr3}, and 
with 
\eq
\Delta \dot{\bPhi}(\tau_{\ell}) := \dot{\bPhi}(\tau_{\ell+1}) - \dot{\bPhi}(\tau_{\ell})
\quad (\ell=1,\dots,L-1),
\eqq
the linear SQR problem  can be restated as
\eqn
\hat{\bmth} :=
\operatorname*{argmin}_{\bmth \in \bbR^{pK}} \bigg\{
 \sum_{\ell=1}^L \sum_{t=1}^n
\rho_{\tau_\ell}(y_t - \bx_t^T \bPhi(\tau_\ell) \bmth) +  \sum_{\ell=1}^{L-1}  c_\ell \,
 \| \Delta \dot{\bPhi}(\tau_\ell) \bmth \|_1 \bigg\},
\label{sqr1}
\eqqn
where $c_\ell := n c w_\ell$. The  weight sequence $\{ w_\ell \}$ is introduced to enable user-specified adjustment to the local smoothness requirement at different quantiles.

Like the ordinary QR problem (Koenker 2005, p.\ 7) and the SQR problem in Li and Megiddo (2026),
the linear SQR problem in (\ref{sqr1}) can be reformulated as a linear program (LP):
\eqn
\lefteqn{
(\hat{\bmgam}, \hat{\bmdel}, \hat{\bu}_1,\dots,\hat{\bu}_L,
\hat{\bv}_1, \dots,\hat{\bv}_L, \hat{\br}_1, \dots,
\hat{\br}_{L-1},\hat{\bs}_1 \dots, \hat{\bs}_{L-1} )  :=} \notag \\
&  &  
\operatorname*{argmin} \
 \bigg\{ 
\sum_{\ell=1}^L ( \tau_\ell \1^T \bu_\ell + (1-\tau_\ell) \1^T \bv_\ell ) + 
\sum_{\ell=1}^{L-1} (
\1^T \br_\ell + \1^T \bs_\ell) \bigg\}
   \notag \\
&& \quad
{\rm s.t.} \quad
\left\{
\begin{array}{l} 
\bX \, \bPhi(\tau_\ell) (\bmgam - \bmdel) + \bu_\ell - \bv_\ell = \by \quad (\ell=1,\dots,L),  \\
c_\ell \, \Delta \dot{\bPhi}(\tau_\ell)  (\bmgam - \bmdel) - (\br_\ell - \bs_\ell) = \0\quad (\ell=1,\dots,L-1),
\end{array} \right.
\label{lp}
\eqqn
where $\bmgam, \bmdel \in \bbR_+^{pK}$, 
$\bu_\ell, \bv_\ell \in \bbR_+^n$, $\br_\ell, \bs_\ell \in \bbR_+^p$ 
are the decision variables, of which $\bu := [\bu_1^T,\dots,\bu_L^T]^T$, $\bv := 
[\bv_1^T,\dots,\bv_L^T]^T$, $\br := [\br_1^T,\dots,\br_{L-1}^T]^T$, and $\bs := [\bs_1^T,\dots,\bs_{L-1}^T]^T$ are auxiliary variables introduced  just for the purpose of linearizing the loss function
of quantile regression. The desired solution in  (\ref{sqr1}) can be written as $\hat{\bmth} = \hat{\bmgam} - \hat{\bmdel}$.

Let $\bmxi  := [\bmgam^T, \bmdel^T, \bu^T, \bv^T,\br^T,\bs^T]^T$. Then, 
the LP problem (\ref{lp}) takes the canonical form
\eqn
\min\{ \bd^T \bmxi \mid \bA \bmxi = \bb; \bmxi \in \bbR_+^{2pK + 2nL+2p(L-1)} \},
\label{primal}
\eqqn
where 
\eqn
\bd  & := & [\0_p^T,\0_p^T, \bc^T, (\1 - \bc)^T, \1_{p(L-1)}^T, \1_{p(L-1)}^T]^T, \notag \\
\bA  & := & \left[
\begin{array}{cccccc}
\bD & -\bD  & \bI_{nL} & - \bI_{nL} & \0 & \0 \\
\bP & -\bP &  \0 & \0 & - \bI_{p(L-1)} & \bI_{p(L-1)}
\end{array}
\right], \notag \\
\bb  & := & [ \1_p^T \otimes \by^T, \0_{p(L-1)}^T]^T, \label{b2}
\eqqn
and
\eqn
\bP  :=  [c_1 \Delta \dot{\bPhi}^T(\tau_1), \dots,c_{L-1} \Delta \dot{\bPhi}^T(\tau_{L-1})]^T. 
\label{P}
\eqqn
The primal problem (\ref{primal}) has a dual that can be written as
\eqn
\max\{ \bb^T \bmlam \mid  \bA^T \bmlam \le \bd; \bmlam \in \bbR^{nL+p(L-1)} \},
\label{dual0}
\eqqn
where $\bmlam$ is the Lagrange multiplier for the equality constraints in (\ref{primal}). 
By a change of variables, the dual in (\ref{dual0}) is transformed into a problem
with bounded variables and equality constraints. This result is stated in the following.

\begin{thm}
{\rm
Let $\bD$,  $\bb$, and $\bP$ be given by  (\ref{D}), (\ref{b2}), and (\ref{P}), respectively.
Define
\eq
\bC  :=  [\bD^T,2 \bP^T]^T, \quad \ba  :=  \sum_{\ell=1}^L (1-\tau_\ell)  \, \bPhi^T(\tau_\ell) \, \bX^T \1_n + \sum_{\ell=1}^{L-1} 
c_\ell \, \Delta \dot{\bPhi}^T(\tau_\ell) \1_p.
\eqq
Then, the dual LP problem (\ref{dual0})  can be rewritten as
\eqn
\max\{ \bb^T \bmzeta \mid \bC^T \bmzeta = \ba;  \bmzeta \in [0,1]^{nL+p(L-1)} \},
\label{dual}
\eqqn
where $\bmlam := \diag\{\bI_{nL} ,2\bI_{p(L-1)}\} \bmzeta
- [  (\1-\bc)^T,\1_{p(L-1)}^T]^T$ is the  Lagrange multiplier for the equality constraints in (\ref{primal}). 
}
\label{thm:sqr1}
\end{thm}

\noindent
{\sc Proof}. By partitioning $\bmlam$ in (\ref{primal}) according to the structure of $\bb$
such that 
\eq
\bmlam := [\bmlam_1^T,\dots,\bmlam_L^T,\bmlam_{L+1}^T,\dots,\bmlam_{2L-1}^T]^T,
\eqq  
the inequalities $\bA^T \bmlam \le \bd$ in (\ref{dual0}) can be written more elaborately as
\eq
\sum_{\ell=1}^L \bPhi^T(\tau_\ell) \, \bX^T \bmlam_\ell +  \sum_{\ell=1}^{L-1} c_\ell \, \Delta \dot{\bPhi}^T(\tau_\ell) \bmlam_{L+\ell}   & \le &  \0, \\
- \sum_{\ell=1}^L \bPhi^T(\tau_\ell) \, \bX^T \bmlam_\ell +  \sum_{\ell=1}^{L-1}c_\ell \, \Delta \dot{\bPhi}^T(\tau_\ell) \bmlam_{L+\ell}  & \le &   \0, 
\eqq
\eq
\bmlam_\ell \le \tau_\ell \1,   \quad  - \bmlam_\ell \le (1-\tau_\ell) \1 \quad (\ell=1,\dots,L), \\
 - \bmlam_{L+\ell} \le  \1, \quad \bmlam_{L+\ell} \le  \1 \quad (\ell=1,\dots,L-1).
\eqq
These inequalities are equivalent to
\eq
\sum_{\ell=1}^L \bPhi^T(\tau_\ell) \, \bX^T \bmlam_\ell + \sum_{\ell=1}^{L-1} c_\ell \, \Delta \dot{\bPhi}^T(\tau_\ell) \bmlam_{L+\ell}  = \0, \\
\bmlam_\ell \in [\tau_\ell-1,\tau_\ell]^n \quad (\ell =1,\dots,L), \\
\bmlam_{L+\ell} \in [-1, 1]^p \quad (\ell =1,\dots,L-1).
\eqq
By a change of variables,
\begin{gather*}
\bmzeta_\ell  := \bmlam_\ell + (1-\tau_\ell) \1 \quad (\ell=1,\dots,L), \\
 \bmzeta_{L+\ell} :=  \frac{1}{2} (  \bmlam_{L+\ell} + \1) \quad (\ell=1,\dots,L-1),
\end{gather*}
we obtain
\eq
\bb^T \bmlam  = \sum_{\ell=1}^L \by^T \bmlam_\ell = \sum_{\ell=1}^L \by^T \bmzeta_\ell  - \sum_{\ell=1}^L (1-\tau_\ell)  
\by^T \1 = \bb^T \bmzeta + \text{constant},
\eqq
\eq
\lefteqn{\sum_{\ell=1}^L \bPhi^T(\tau_\ell) \, \bX^T \bmlam_\ell + \sum_{\ell=1}^{L-1} c_\ell  \,\Delta \dot{\bPhi}^T(\tau_\ell) \bmlam_{L+\ell} } \\
& = & \sum_{\ell=1}^L  \bPhi^T(\tau_\ell) \, \bX^T  \bmzeta_\ell +  \sum_{\ell=1}^{L-1} 2 c_\ell \, \Delta \dot{\bPhi}^T(\tau_\ell) \bmzeta_{L+\ell}  \\
& & - \sum_{\ell=1}^L (1-\tau_\ell) \, \bPhi^T(\tau_\ell) \, \bX^T \1
 - \sum_{\ell=1}^{L-1}c_\ell  \, \ddot{\bPhi}^T(\tau_\ell) \1, 
\eqq
\eq
\bmlam_\ell \in [\tau_\ell-1,\tau_\ell]^n & \leftrightarrow & \bmzeta_\ell \in [0,1]^n \quad (\ell=1,\dots,L), \\
\bmlam_{L+\ell} \in [-1,1]^p & \leftrightarrow & \bmzeta_{L+\ell} \in [0,1]^p  \quad (\ell=1,\dots,L-1).
\eqq
Combining these expressions proves the assertion. \qed

Using the quantities in the dual problem (\ref{dual}), the primal problem (\ref{primal}) 
is restated in the following theorem in terms of the primary variable $\bmth$.

\begin{thm}
{\rm 
With $\bC$ and $\ba$ given by Theorem~\ref{thm:sqr1} and with $\bb$ defined by (\ref{b2}),
the primal LP problem (\ref{primal}) can be rewritten as
\eqn
\min\{ \ba^T \bmth \mid \bC \bmth + \bz - \bw = \bb; \bmth \in \bbR^{pK}; \bz,\bw \in \bbR_+^{nL+p(L-1)}\},
\label{primal2}
\eqqn
where $\bmth := \bmgam-\bmdel$, $\bz := [\bu^T,2\bs^T]^T$,
and $\bw := [\bv^T,2\br^T]^T$.
}
\end{thm}

\noindent
{\sc Proof}. Observe that the equality constraints in (\ref{primal}) can be written as 
\eqn
\bG \bmth + \bz - \bw = \bb.
\label{ec}
\eqqn
Under these constraints, $\bv_\ell  = \bX \, \bPhi(\tau_\ell) \bmth + \bu_\ell - \by$ and $ \br_\ell = 
c_\ell \, \Delta \dot{\bPhi}(\tau_\ell) \bmth + \bs_\ell$. Therefore,
\eq
\bd^T \bmxi  & = & \sum_{\ell=1}^L \{ \tau_\ell \1^T \bu_\ell + (1-\tau_\ell) \1^T \bv_\ell \} + \sum_{\ell=1}^{L-1} \{ \1^T \br_\ell + \1^T \bs_\ell \} \\
&=&  \sum_{\ell=1}^L \{ \tau_\ell \1^T \bu_\ell +  (1-\tau_\ell) \1^T  ( \bX \, \bPhi(\tau_\ell) \bmth +  \bu_\ell - \by) \} 
+ \sum_{\ell=1}^{L-1} \{ \1^T (c_\ell \, \Delta \dot{\bPhi}(\tau_\ell) \bmth + \bs_\ell) + \1^T \bs_\ell \} \\
& = &  \sum_{\ell=1}^L (1-\tau_\ell) \1^T  \bX \, \bPhi(\tau_\ell) \bmth + \sum_{\ell=1}^{L-1} 
c_\ell \1^T \Delta \dot{\bPhi}(\tau_\ell) \bmth  
+ \sum_{\ell=1}^L \{ \1^T \bu_\ell - (1-\tau_\ell)  \1^T  \by  \} 
+\sum_{\ell=1}^{L-1}  2 \1^T \bs_\ell   \\
& = & \ba^T \bmth +  \| \bz \|_1 + \text{constant}.
\eqq
The assertion follows upon noting that $\bz$ is not a free parameter.under the condition (\ref{ec}) 
because it is equivalent to $\bC \bmth - \bw \le \bb$ and $\bz := \bb - (\bC \bmth - \bw)$. \qed.

Given the canonical form (\ref{primal}),  the linear SQR solution can be computed by any general-purpose LP solvers. For example, the {\tt lp} function in the R package `lpSolve'
(Berkelaar 2024)  provides an interface to the open-source software {\tt lp\_solve} for solving LP problems by a simplex method. However, in using such solvers, the large number of decision variables and constraints in (\ref{primal})  can be a challenge to both computer memory and computer time.

A more efficient alternative is the primal-dual interior-point method developed by Portnoy and Koenker (1997) specifically for quantile regression. This method solves the primal-dual pair jointly by  Newton's method in which the positivity constraints are enforced by log barriers (Koenker and Ng 2005).  An implementation of the method in FORTRAN, called {\tt rqfnb.f}, is available in the `quanreg' package (Koenker 2005).  The {\tt rq.fit.fnb} function in this package provides an interface to the FORTRAN code for solving the ordinary QR problem. The input variables accepted by this function are $\by$, $\bX$, and $(1-\tau) \bX^T \1_n$ according to the dual 
formulation of the QR problem
\eqn
\max \{ \by^T \bmzeta \mid \bX^T \bmzeta = (1-\tau) \bX^T \1; \bmzeta \in [0,1]^n\}.
\label{qr:dual}
\eqqn
A comparison of (\ref{qr:dual}) with (\ref{dual}) suggests that the linear SQR problem can be solved by a modified 
version of this function to accept $\bb$, $\bC$, and $\ba$, respectively, as input variables. In addition, 
the initial value of  $\bmzeta$ need be changed from $1-\tau$ to  $[(1-\tau_1) \1_n^T,\dots,(1-\tau_L) \1_n^T,0.5 \1_{p(L-1)}^T]^T$. The modified version, called {\tt rq.fit.fnb2}, is used in our R implementation,

The R function {\tt rq.fit.sfn} in the `quantreg' package is an alternative implementation 
of the interior-point method of Portnoy and Koenker (1997). It employs sparse-matrix representations and therefore consumes less memory space than  {\tt rq.fit.fnb2} for large problems. Although not as fast as {\tt rq.fit.fnb2}
for small problems based on our experience, its greater flexibility on input variables allows {\tt rq.fit.sfn} to be invoked without modification 
to compute the linear SQR solution based on the input variables $\bb$, $\bC$, and $\ba$.

\subsection{Smoothing Parameter Selection}

To device a data-driven criterion for selecting the smoothing parameter $c$, we extend the method  used by 
Koenker et al.\  (1994).  Under a location-scale model  $y_t= \mu(x_t) + \sig z_t$ with
$\{ z_t \} \sim \text{i.i.d. }   h(z) $ and $h(\cdot) \propto \exp(-\rho_\tau(\cdot))$, Koenker et al.\  (1994)   (c.f.\ Koenker 2005, p.\ 234)
select the smoothing parameter in a smoothing spline estimator $\hat{\mu}(\cdot)$ of $\mu(\cdot)$ 
by employing Schwarz's BIC (or SIC) criterion  of the form 
$2n \log (\hat{\sig}) + \log(n) m$, where $\hat{\sig} := n^{-1} \sum_{t=1}^n \rho_\tau(y_t -\hat{\mu}(x_t))$ 
and  $m :=  \sum_{t=1}^n I(|y_t -\hat{\mu}(x_t) | < \ep)$ (for small $\ep$).
The choice of complexity measure $m$ is motivated by the fact that in the case of linear quantile regression with $p$ explanatory variables, the QR solution interpolates the $y_t$'s at exactly $p$ points so that $m = p$.

 Inspired by this idea, we consider 
the following criterion
for selecting the parameter $c$ in both cubic and linear SQR solutions:
\eqn
{\rm BIC}(c) :=  2 n \log\bigg( L^{-1} \sum_{\ell=1}^L \hat{\sig}(\tau_\ell) \bigg)
+ \log(n) \, \bigg( L^{-1} \sum_{\ell=1}^L m(\tau_\ell) \bigg),
\label{BIC}
\eqqn
where $\hat{\sig}(\tau_\ell) := n^{-1} \sum_{t=1}^n \rho_{\tau_\ell} (y_t - \bx_t^T \hat{\bmbeta}(\tau_\ell))$ and $m(\tau_\ell) := \sum_{t=1}^n I(|y_t -\bx_t^T \hat{\bmbeta}(\tau_\ell) | < \ep)$. 
Intuitively, as $c$ increases, the smoothing effect makes  $\hat{\sig}(\tau_\ell)$ increase and 
$m(\tau_\ell)$ decrease.  The best choice of $c$  strikes a balance between these opposing movements.
By convention,  Akaike's AIC criterion can be obtained by replacing $\log(n)$  in (\ref{BIC}) with te constant 2. When $n$ is sufficiently large ($n \ge 8$ so that $\log(n) > 2$), the heavier penalty in BIC tends to yield a smoother estimate. 

In the R implementation, the smoothing parameter $c$ is reparameterized by {\tt spar} in a  way similar to the R function {\tt smooth.spline} (R Core Team 2024),  i.e., $c := r \times 1000^{{\tt spar}-1}$, where
for the cubic SQR problem, 
\eq
r := \frac{n^{-1}  \sum_{\ell=1}^L \| \bX \bPhi(\tau_\ell)\|_1}
{ \sum_{\ell=1}^L w_\ell \, \tr( \ddot{\bPhi}^T(\tau_\ell) \, \ddot{\bPhi}(\tau_\ell))}.
\eqq
and for the linear SQR problem, 
 \eq
r := \frac{n^{-1}  \sum_{\ell=1}^L \| \bX \bPhi(\tau_\ell)\|_1}
{ \sum_{\ell=1}^{L-1} w_\ell \|  \Delta \dot{\bPhi}(\tau_\ell) \|_1}.
 \eqq 
The AIC and BIC criteria are minimized as functions of {\tt spar} in a suitable interval.

\subsection{Bootstrap Confidence Band}

A pointwise bootstrap confidence band can be constructed from the SQR estimates using various sampling schemes (Koenker 2005). In particular, we consider the so-called $(x,y)$-pair method discussed in Koenker (2005, p.\ 107). This simple but effective method generates bootstrap samples $\{(\bx_t^{(b)},y_t^{(b)}): t=1,\dots,n\}$ 
$(b=1,\dots,B)$ by  drawing from the observed pairs $(\bx_t,y_t)$ $(t=1,\dots,n)$ randomly with replacement.
Using these data sets, we obtain the bootstrap SQR estimates $\hat{\bmbeta}^{(b)}(\cdot)$ $(b=1,\dots,B)$. A 90\% confidence band for $\bmbeta_0(\tau)$ at fixed $\tau$ is  given by the interval with the 5th and 95th percentiles of  $\{\hat{\bmbeta}^{(b)}(\tau): b=1,\dots,B\}$ as lower and upper limits.
 
 The $(x,y)$-pair method can be extended to handle serially-dependent time series data. Instead 
 of sampling the $(x,y)$-pairs individually, one draws random blocks of consecutive $(x,y)$-pairs of fixed length 
 from the series $\{ (\bx_t,y_t): t=1,\dots,n\}$ and concatenates these blocks to create a bootstrap sample 
$ \{(\bx_t^{(b)},y_t^{(b)}): t=1,\dots,n\}$. This block bootstrap method preserves certain serial dependence of the original data without relying on a parametric model, which can be difficult to have especially when the entire conditional distribution 
need to be represented faithfully. The block bootstrap method  reduces to the original $(x,y)$-pair method when the block length is set to 1. Comprehensive reviews of additional bootstrap methods for time series can be found in B\"uhlmann (2002) and Lahiri (2003).

\subsection{Derivatives}

Because both cubic and linear SQR estimates of $\bmbeta_0(\cdot)$ take the form 
$\hat{\bmbeta}(\cdot) = \bPhi(\cdot)\,\hat{\bmth}$, one can easily compute the derivatives 
of these estimates using $\hat{\dot{\bmbeta}}(\cdot) := \dot{ \bPhi}(\cdot)\,\hat{\bmth}$, where
$\dot{ \bPhi}(\cdot)$ is defined by the derivatives of the spline basis functions in the same way as
$\bPhi(\cdot)$ is defined by the spline basis functions.
For cubic SQR, the derivatives are continuous functions of $\tau$. For linear
SQR, the derivatives are piecewise-constant functions of $\tau$ with possible discontinuity
or jump at some points in $\{ \tau_\ell\}$. 
The bootstrap method discussed in the proceeding section can be extended straightforwardly to produce a confidence band for the derivatives. Having the ability to offer an assessment for the derivatives 
of the regression coefficients  is an added benefit  of the SQR method in comparison with QR.
Based on  the estimated derivatives, one can derive an estimate of the form 
$\bx^T \hat{\dot{\bmbeta}}(\tau)$  for the conditional quantile density
function $(d /d\tau)  F^{-1}(\tau \mid \bx) := \bx^T \dot{\bmbeta}_0(\tau)$.

\section{Simulation}

Besides offering a concise representation of the functional regression coefficients,
an important premise of the SQR method is its capability of producing a more accurate
estimate of these coefficients than that  by the QR method when the underlying functions vary smoothly across quantiles. Our  simulation study serves to evaluate  this capability.

The first experiment employs the linear model used by Koenker (2005, p.\ 158), i.e.,
\eqn
y_t = a_0 + a_1 x_{1,t} + a_2 x_{2,t} + \ep_t \quad (t=1,\dots,n),
\label{lin}
\eqqn
where $x_{1,t} := q(t/(n+1))$ and $x_{2,t} := |q(t/(n+1))|$, with $q(\cdot)$ 
being the standard normal quantile function and $\{ \ep_t \}$ being
an i.i.d.\ sequence of standard normal random variables. We set $a_0 :=1.0$, $a_1 := 2.0$, and 
$a_2 := 1.5$. The conditional quantile function
of $y_t$ given $\bx_t := [1,x_{1,t},x_{2,t}]^T$ takes the form (\ref{F}) with 
$\bmbeta_0(\cdot) := [a_0 + q(\cdot), a_1, a_2]^T$. In this example, the first element of $\bmbeta_0(\cdot)$
is a nonlinear function of $\tau$ and the  remaining elements are constants that do not depend on $\tau$.

To evaluate the estimation accuracy, we employ the total mean absolute error
\eqn
\text{MAE} := E \bigg\{ \frac{1} {|\cT| } \sum_{\tau \in \cT}  \|\hat{\bmbeta}(\tau) - \bmbeta_0(\tau)\|_1  \bigg\},
\label{mae}
\eqqn
where $\cT \subset [a,b]$ is a set of quantile levels for evaluation, with $|\cT|$ denoting the cardinality of $\cT$.
Figure~\ref{fig:lin} shows the MAE from 2000 Monte Carlo runs of linear and cubic SQR estimates of  $\bmbeta_0(\cdot)$ as a function of the smoothing parameter {\tt spar}. The estimates are obtained by solving (\ref{sqr}) with $\{ \tau_\ell \}
:= \{ 0.04,0.06,\dots,0.96\}$ and the MAE  in (\ref{mae})  is evaluated with $\cT := \{ \tau_\ell \}$.

\begin{figure}[t]
\centering
\includegraphics[height=2.75in,angle=-90]{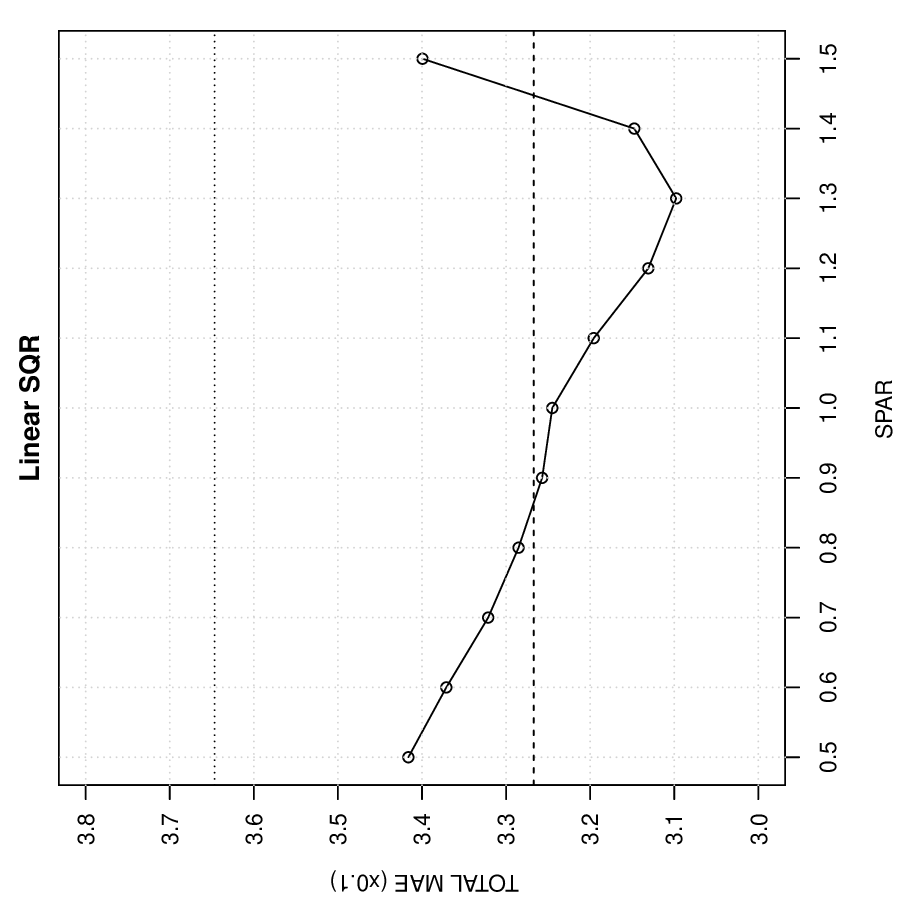}  
\includegraphics[height=2.75in,angle=-90]{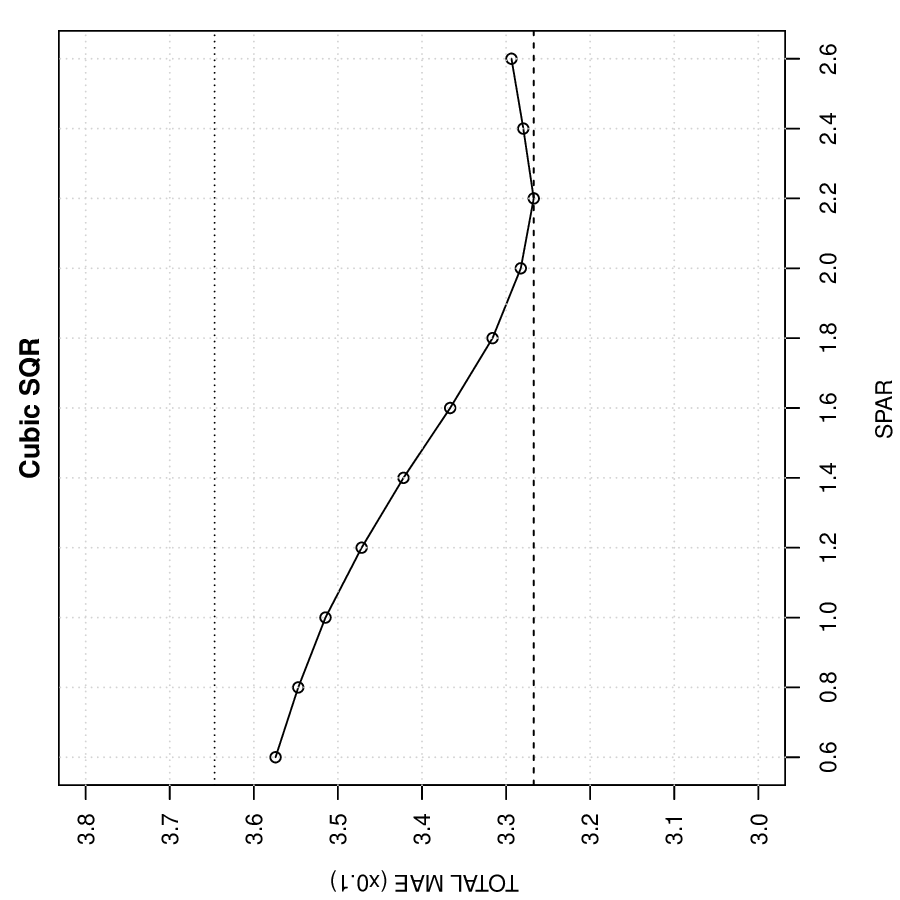}  
\caption{Total mean absolute error (MAE) of linear SQR estimates  (left) and cubic SQR estimates (right) against  
smoothing parameter {\tt spar}  for the regression model  (\ref{lin}) with $n=200$.
Horizontal dotted line represents the error of QR estimates. 
Horizontal dashed line represents the best error achieved by the SQR estimator of Li and Megiddo (2026). 
Results are based on 2000 Monte Carlo runs.}
 \label{fig:lin}
\end{figure}

As can be seen in Figure~\ref{fig:lin}, both linear and cubic SQR outperform QR  (depicted by horizontal dotted  line) for a range of {\tt spar} values not too small or too large. Obviously, the estimates of the second and third elements in $\bmbeta_0(\cdot)$ benefit from a large {\tt spar} for noise reduction, but when {\tt spar} is too large,  the  excessive bias in the estimate of the first element becomes dominant, resulting in an increase of MAE.
A similar behavior is observed by Koenker (2005, p.\ 158) regarding the effect of the bandwidth parameter in the kernel smoother.

The linear SQR shown in the left panel of Figure~\ref{fig:lin}
outperforms the SQR  of Li and Megiddo (2026) (horizontal dashed  line) in terms of the best MAE, thanks mainly to the replacement of cubic splines by linear splines which are 
better suited for the second and third elements in $\bmbeta_0(\cdot)$. 
The cubic SQR shown in the right panel of Figure~\ref{fig:lin} is comparable 
to the SQR of Li and Megiddo (2026) in terms of the best MAE,
so the replacement of $\ell_1$-norm with $\ell_2$-norm in the penalty makes no noticeable  difference in this example.

\medskip

The second experiment is based on a quantile autoregressive (QAR) model for time series 
(Koenker 2005, p.\ 262), i.e.,
\eqn
y_t = a_0(u_t) + a_1(u_t) y_{t-1} \quad (t=1,\dots,n),
\label{sim2}
\eqqn
where $\{ u_t \}$ is an i.i.d.\  sequence of $U(0,1)$ random variables, $a_0(\tau) := 0.1 q(\tau)$ (nonlinear) and $a_1(\tau)  := 0.85+0.1 \tau + 0.25 (\tau-0.5) I(\tau > 0.5)$ (piecewise linear).
The conditional quantile function of $y_t$ given $\bx_t := [1,y_{t-1}]^T$ takes the form (\ref{F})
with $\bmbeta_0(\cdot) := [ a_0(\cdot),a_1(\cdot)]^T$.

\begin{figure}[p]
\centering
\includegraphics[height=2.75in,angle=-90]{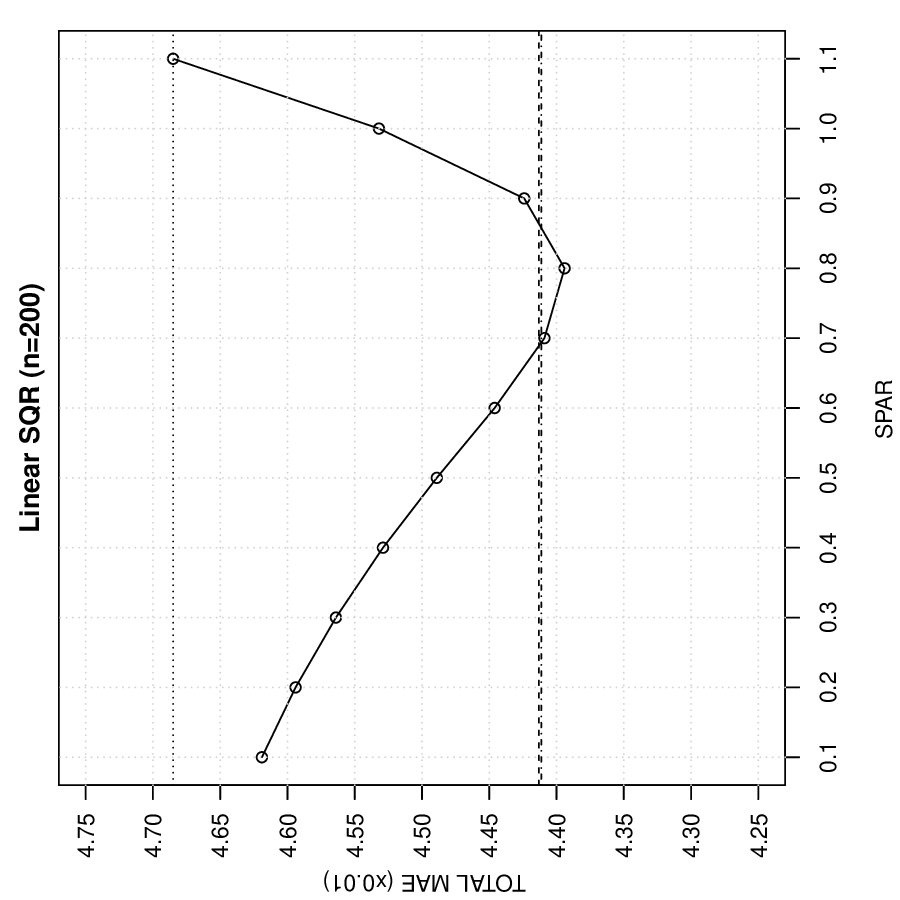} 
\includegraphics[height=2.75in,angle=-90]{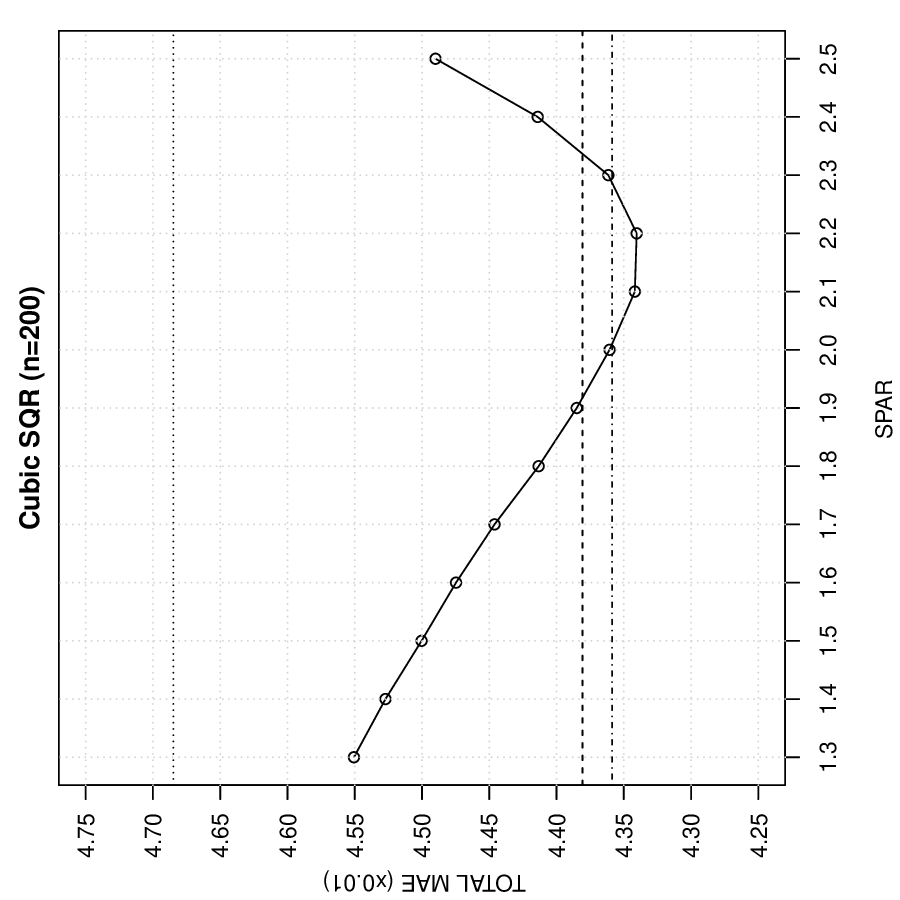}  
\\
\centering{\small (a)} \\
\vspace{-0.2in}
\includegraphics[height=2.75in,angle=-90]{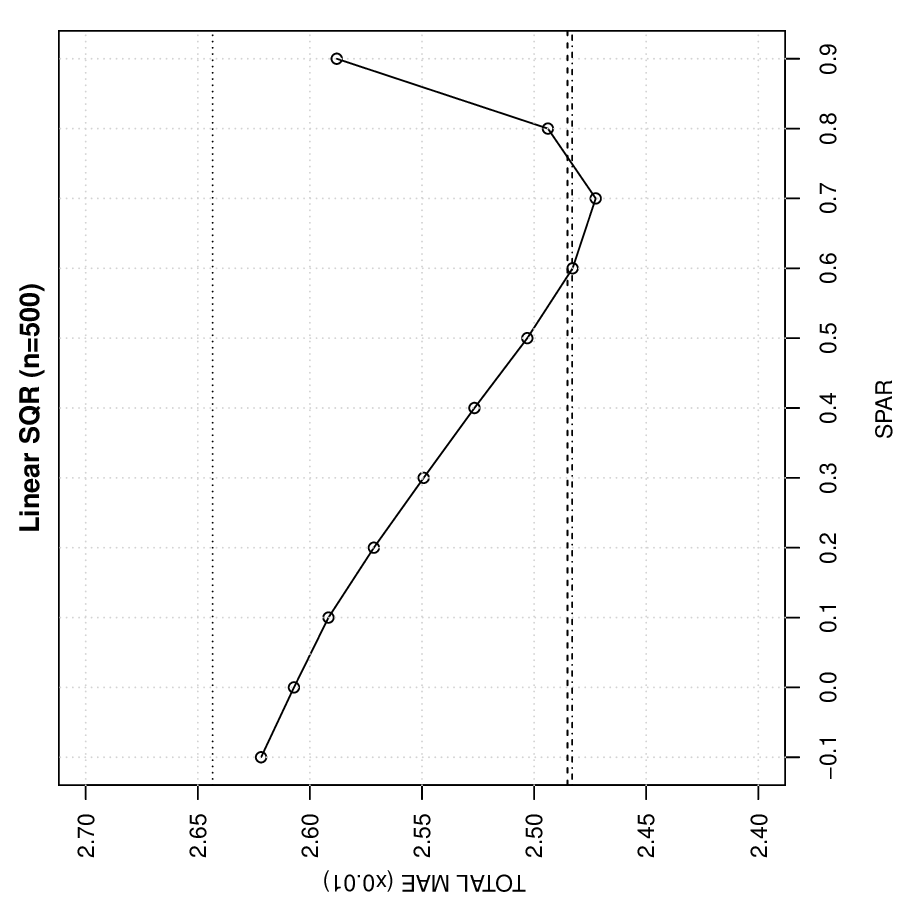} 
\includegraphics[height=2.75in,angle=-90]{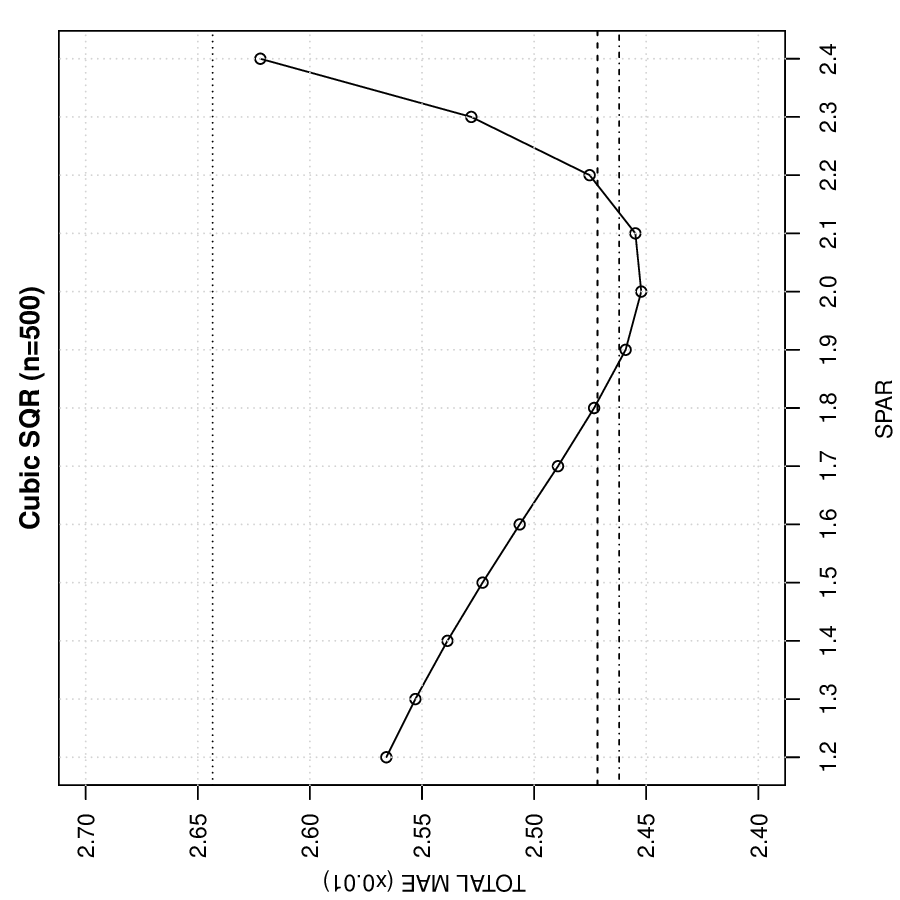}
  \\
\centering{\small (b)} 
\caption{Total mean absolute error (MAE) of linear SQR estimates  (left) and cubic SQR estimates (right) against  
smoothing parameter {\tt spar} for the QAR model  (\ref{sim2}) with (a)  $n=200$ and (b) $n=500$. 
Horizontal dotted line represents the error of QR estimates. Horizontal dashed line and
dash-dotted line represent the error of SQR estimates with {\tt spar} selected by AIC and BIC,
respectively. Results are based on 2000 Monte Carlo runs.} \label{fig:sim2}
\end{figure}

Figure~\ref{fig:sim2} shows the total MAE, as a function of {\tt spar}, of the linear
and cubic SQR estimates of $\bmbeta_0(\cdot)$  
based on 2000 Monte Carlo runs with $n=200$ and $n=500$, calculated 
with $\cT := \{ \tau_\ell\} := \{0.05,0.07,\dots,0.95\}$.
The result in Figure~\ref{fig:sim2} confirms again that the SQR estimates are able to improve 
the QR estimates (depicted by horizontal dotted line) for a range of {\tt spar} values. In addition,
both AIC and BIC (horizontal dashed line and dash-dotted line) 
work reasonably well for automatic selection of {\tt spar}, 
with BIC yielding slightly better estimates especially for the cubic SQR.
Unlike the previous example, the best MAE of the cubic SQR is smaller than that of the linear SQR in this example.

It is worth pointing out that although AIC and BIC produce reasonably good results in this example,
these criteria are not necessarily smooth as functions of the smoothing parameter. Figure~\ref{fig:aic}
provides an example of these functions. Therefore, it is possible that a line search algorithm produces
a local minimum instead of a global one.

\begin{figure}[t]
\centering
\includegraphics[height=5.5in,angle=-90]{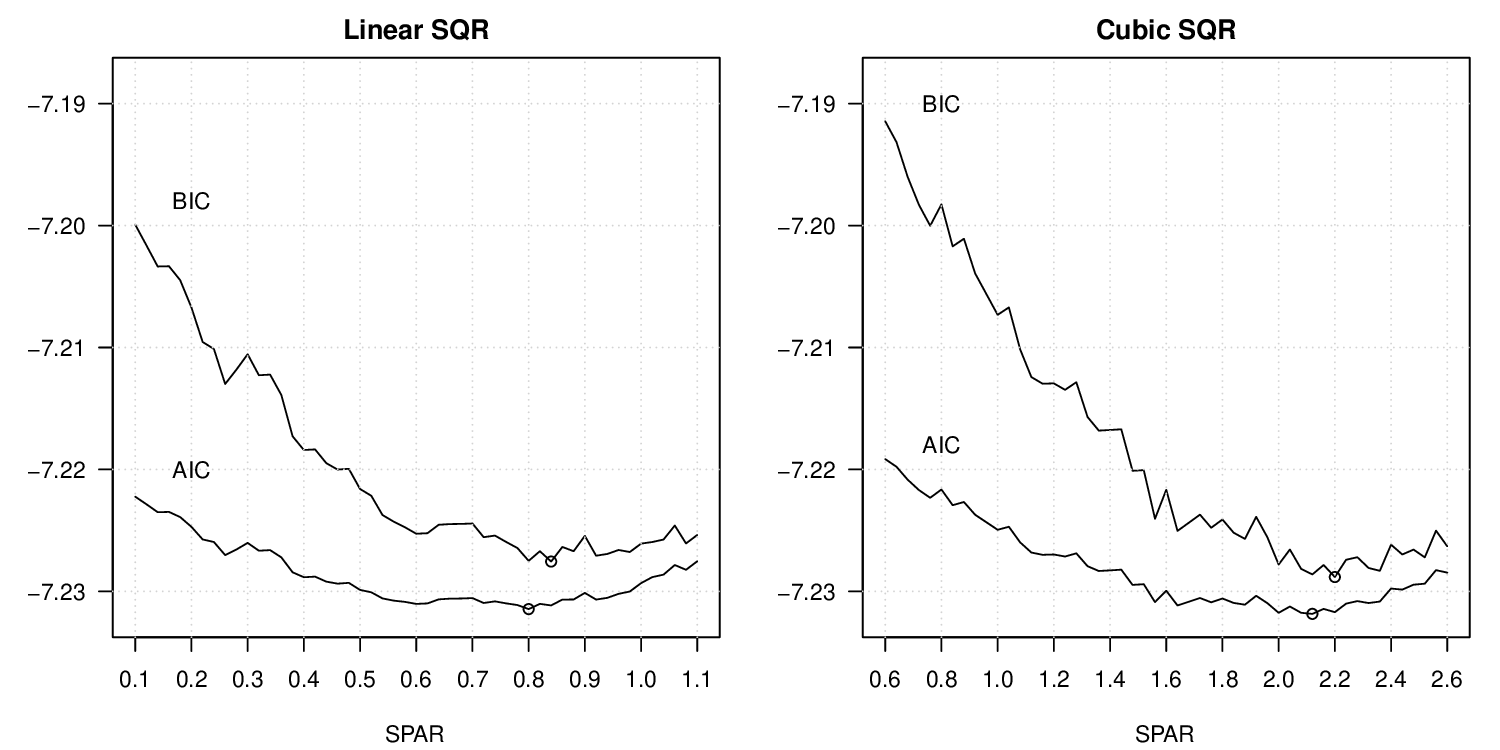} 
\caption{Plot of AIC and BIC as functions of {\tt spar} for a simulated data set from (\ref{sim2}). 
Circle identifies the global minimum.} 
\label{fig:aic}
\end{figure}

\begin{table}[t]
\begin{center} 
\caption{Mean Absolute Error  of SQR Estimates for Coefficients in (\ref{sim2})} \label{tab:err}
{\small
\begin{tabular}{cc|ccccc|ccccc} \hline
& & \multicolumn{5}{c|}{MAE of $\hat{a}_0(\tau)$ (x0.001)} & \multicolumn{4}{c}{MAE of $\hat{a}_1(\tau)$ (x0.01)} \\
$n$  & $\tau$ & QR & SQR$^a$ & SQR$^b$ &SQR$^c$ &  QR-S & QR &  SQR$^a$ &  SQR$^b$ & SQR$^c$  & QR-S \\  \hline
 200 & 0.25  & 9.170 & 8.737 & {\bf 8.563}  & 8.751 & 9.104 & 3.603 & {\bf 3.305} & 3.319 & 3.359 &  3.572  \\
        & 0.50  & 8.411 & 7.845 & {\bf 7.719}  & 8.049 & 8.364 & 3.014 & 2.753         & {\bf 2.659} & 2.831 &  3.000 \\
        & 0.75  & 9.506 & 9.174 & {\bf 9.011}  & 9.213 & 9.457 & 3.650 & {\bf 3.423} & 3.443 &  3.479 & 3.637 \\
 500 & 0.25 & 5.080 & 5.007 & {\bf 4.937}  & 4.986 & 5.051 & 1.983 & {\bf 1.830}  & 1.839  & 1.864 & 1.956  \\
        & 0.50 & 4.697 & 4.407 & {\bf 4.395}  & 4.492 & 4.669 & 1.744 & 1.661         & {\bf 1.624}  & 1.671 & 1.733 \\
        & 0.75 & 5.461 & 5.229 & {\bf 5.192}  & 5.263 & 5.425 & 1.986 & {\bf 1.855} & 1.871  & 1.893 & 1.979  \\
\hline
\end{tabular} 
}
\end{center}
{\footnotesize 
\begin{center}
\begin{minipage}{6in}
$a$, linear SQR, {\tt spar} = 0.9 $(n=200)$ and 0.8 ($n=500$); 
$b$, cubic SQR, {\tt spar} = 2.3 $(n=200)$ and 2.2 ($n=500$); 
$c$, SQR of Li and Megiddo (2026), {\tt spar} = 0.7 $(n=200)$ and 0.6 ($n=500$); QR-S, post-smoothing of QR estimates by {\tt smooth.spline}. Boldface font highlights the smallest error in each case.
Results are based on 2000 Monte Carlo run. 
\end{minipage}
\end{center}
}
\end{table}

While  aimed mainly at estimating $\bmbeta_0(\cdot)$ as a function, 
the SQR method  is also expected to  improve the point estimation at a single quantile of interest
under smoothness conditions. This added benefit  is confirmed by the  result  in Table~\ref{tab:err}, which  contains the mean absolute error of the QR and SQR estimates for $a_0(\tau)$ and $a_1(\tau)$ in (\ref{sim2}) at fixed quantile levels $\tau=0.25$, $0.5$, and $0.75$, calculated from 2000 Monte Carlo runs with $n=200$ and $500$. 
For comparison, Table~\ref{tab:err} also contains the result of  the SQR of Li and Megiddo (2026) 
and the result of a post-smoothing method, labeled QR-S, where the point estimates are obtained by smoothing  the QR estimates obtained at $\{ \tau_\ell\}$ using the R function {\tt smooth.spline} (R Core Team 2024). 

According to Table~\ref{tab:err}, the linear and cubic SQR  not only outperform 
QR and QR-S but also outperform the SQR of  Li and Megiddo (2026)
at each quantile level of interest in terms of the MAE.
Moreover, the cubic SQR offers the best MAE for $a_0(\tau)$ at each of the three quantiles, whereas the linear SQR offers the best MAE for $a_1(\tau)$ at $\tau=0.25$ and $0.75$ where the true function is linear in the neighborhood.

\medskip

The third experiment is intended  to demonstrate the usefulness of an interpolation 
method based on SQR.  In instead of solving the SQR problem  (\ref{sqr}) with $\{ \tau_\ell \} := \cT$,
the interpolation  method solves it  with a smaller set of quantile levels, 
i.e., $\{ \tau_\ell \} \subset \cT$, and then interpolates the fitted values  using the spline basis in (\ref{betah})
to obtain the remaining values on $\cT$. This method not only reduces the computational burden but may improve the estimation accuracy on $\cT$ when compared to the estimates obtained with the full set of quantiles, i.e., $\{ \tau_\ell \} := \cT$. To validate the last point, we employ a random-coefficient regression model
\eqn
y_t = a_0(u_t) + a_1(u_t) x_t \quad (t=1,\dots,n),
\label{sim}
\eqqn
where $a_0(\tau) := 0.1 q(\tau)$ (nonlinear), $a_1(\tau)  := 0.85+0.1 \tau^2 + (\tau-0.5)^2 I(\tau > 0.5)$ 
(piecewise quadratic), $\{ u_t \}$ is an i.i.d.\ sequence  of $U(0,1)$ random variables, 
and $\{ x_t \}$ is an i.i.d.\ sequence of $U(0,5)$ random variables, which is independent of $\{ u_t \}$. 
In this case, the conditional quantile function of $y_t$ given $\bx_t := [1,x_t]^T$ 
satisfies (\ref{F}) with $\bmbeta_0(\cdot) := [a_0(\cdot),a_1(\cdot)]^T$. 

\begin{figure}[p]
\centering
\includegraphics[height=2.75in,angle=-90]{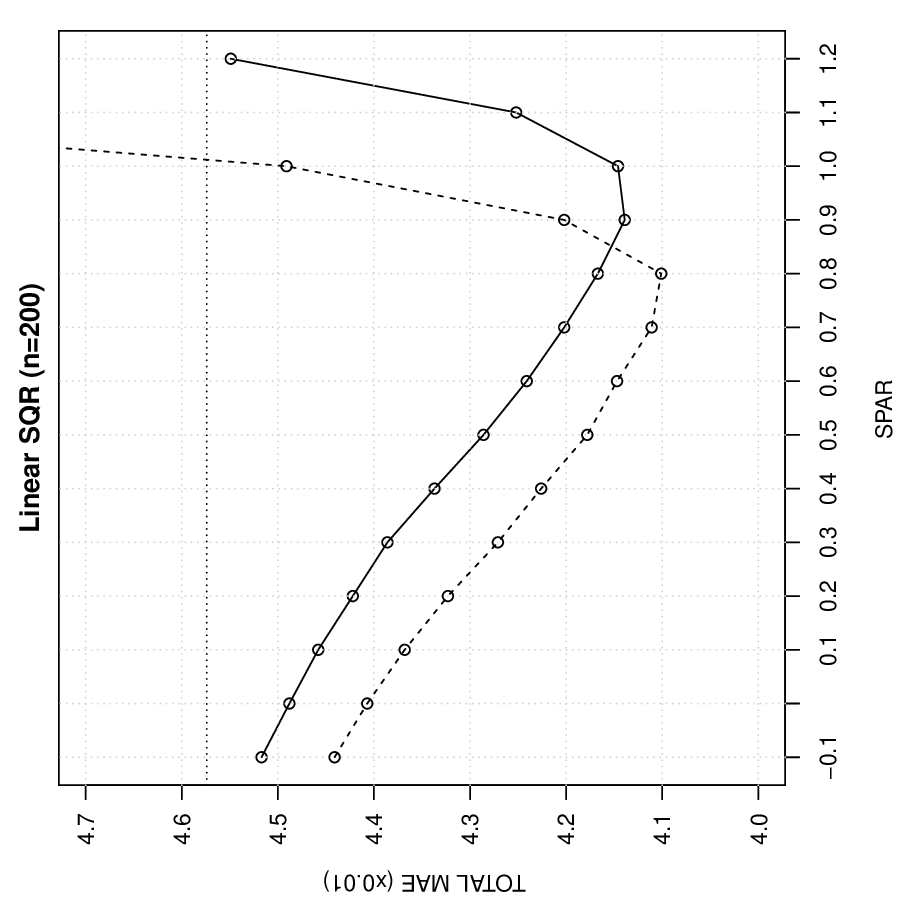} 
\includegraphics[height=2.75in,angle=-90]{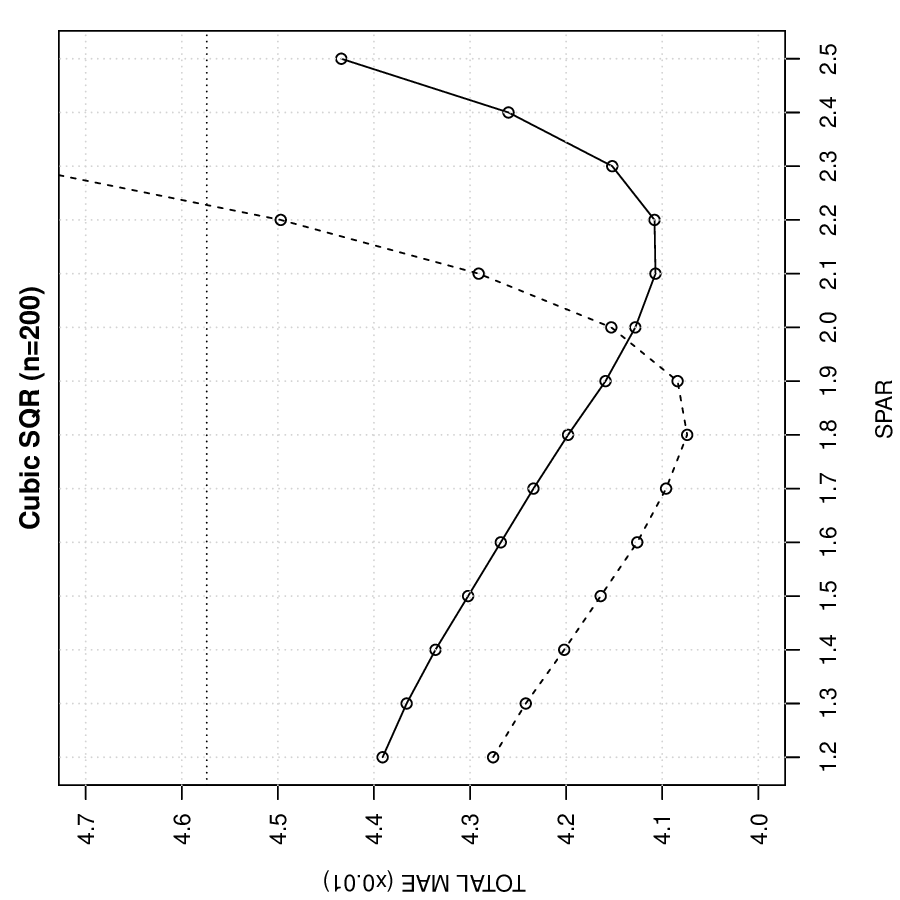}  
\\
\centering{\small (a)} \\
\vspace{-0.2in}
\includegraphics[height=2.75in,angle=-90]{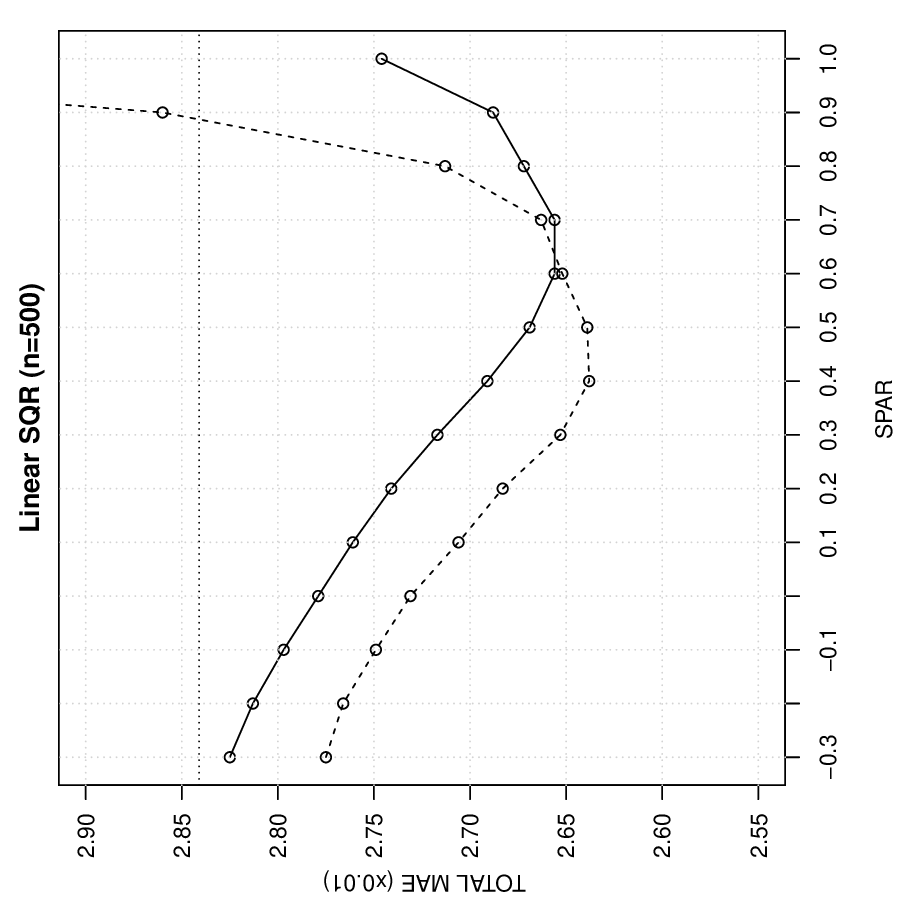} 
\includegraphics[height=2.75in,angle=-90]{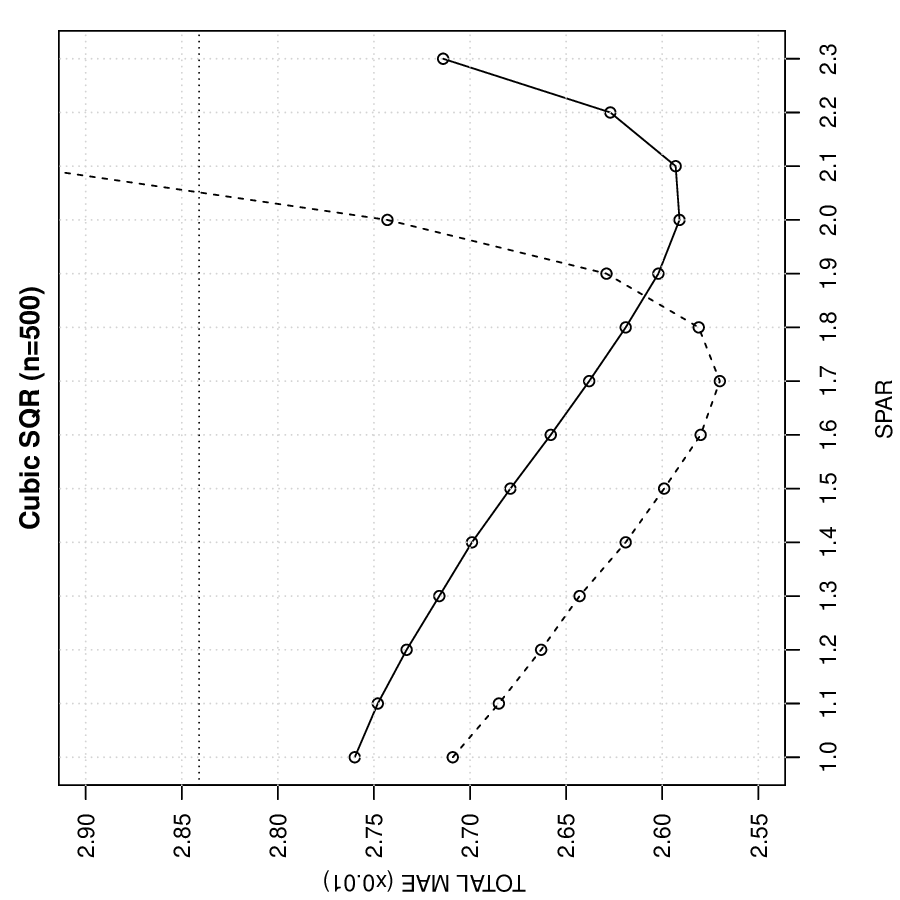}
  \\
\centering{\small (b)} 
\caption{Total mean absolute error (MAE) of linear SQR estimates  (left) and cubic SQR estimates (right)
with different values of smoothing parameter {\tt spar} 
from simulated data by (\ref{sim}) for (a) $n=200$ and (b) $n=500$ with MAE 
evaluated at $\cT := \{ 0.04,0.06,\dots,0.96\}$. 
Solid line,  estimates obtained with $\{ \tau_\ell \} := \{ 0.04,0.06,\dots,0.96\} = \cT$. 
Dashed line, estimates obtained with $\{ \tau_\ell \} := \{ 0.04,0.08,\dots,0.96\} \subset \cT$. 
Horizontal dotted line represents the error of QR estimates.
Results are based on 2000 Monte Carlo runs.} \label{fig:sim}
\end{figure}

Figure~\ref{fig:sim} depicts the MAE calculated from 2000 Monte Carlo runs on $\cT := \{ 0.04,0.06,\dots,$ $0.96\}$ with different values of {\tt spar} for the linear and cubic SQR estimates. The SQR estimates are obtained 
by solving (\ref{sqr})  respectively with the full set $\{ \tau_\ell \} := \cT$ (solid line) and  the subset $\{ \tau_\ell \} := \{ 0.04,0.08,\dots,0.96\} \subset \cT$ (dashed line). Besides confirming the superiority of SQR over  QR (depicted by horizontal dotted line), the result in Figure~\ref{fig:sim} shows that the best MAE of the SQR estimates obtained
by solving (\ref{sqr}) with the subset  is indeed smaller than the best MAE of their counterparts obtained 
by solving (\ref{sqr}) with  the full set. 
Such result occurs when the statistical variability 
of the interpolated values becomes smaller than that of the fitted values obtained using  the full set and 
outweighs the increased bias as a result of interpolation. Needless to say, if the subset   is too coarse for the dynamics of the functional coefficients, the  bias in the interpolated estimates may render them
less accurate than the fitted values obtained using the full set.

\section{Application}

Having demonstrated the ability of the SQR method to produce more accurate estimates
than the QR method using simulated data, we now present two real-data examples  in  this section to demonstrate the application of the method. 

\begin{figure}[p]
\centering
\includegraphics[height=5.5in,angle=-90]{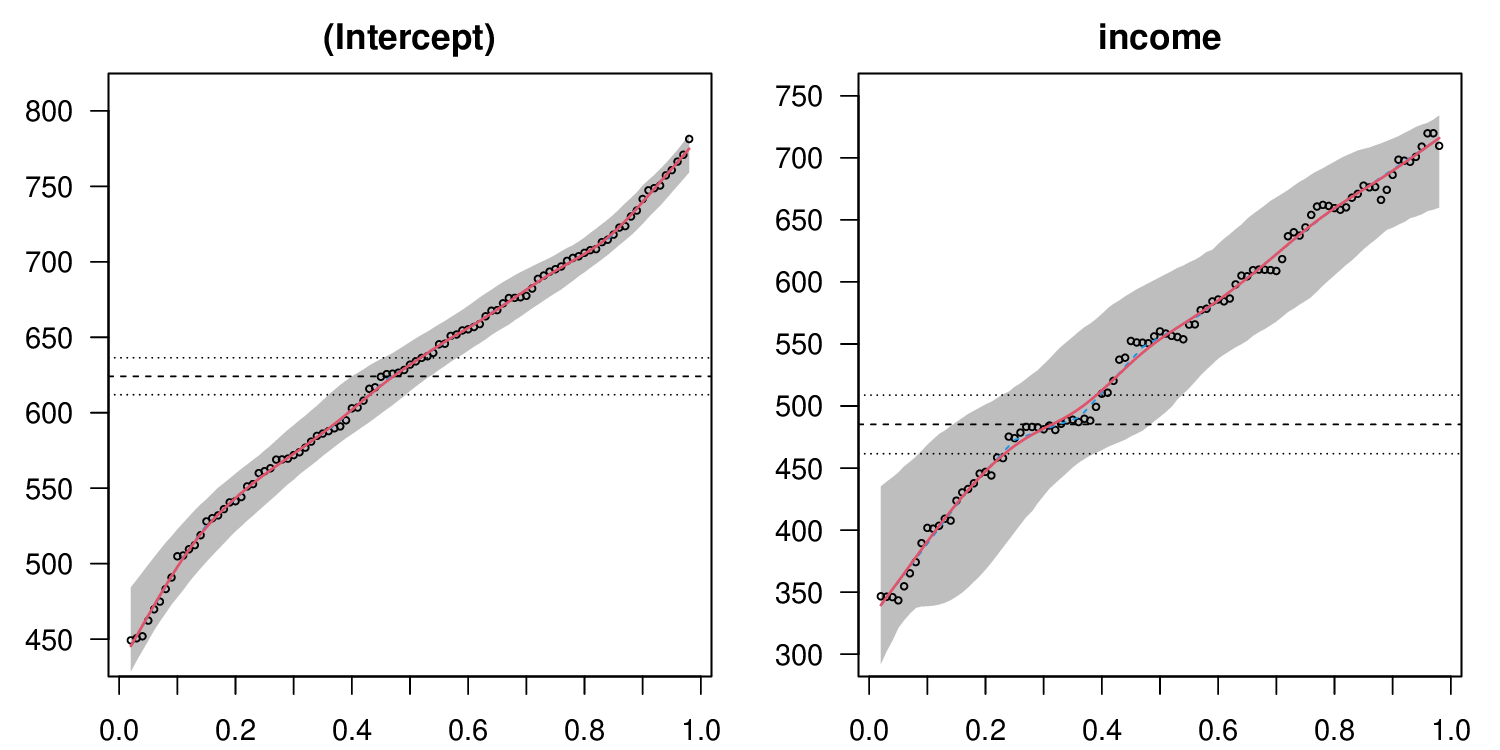} \\
{\small (a)} \\
\vspace{-0.2in}
\includegraphics[height=5.5in,angle=-90]{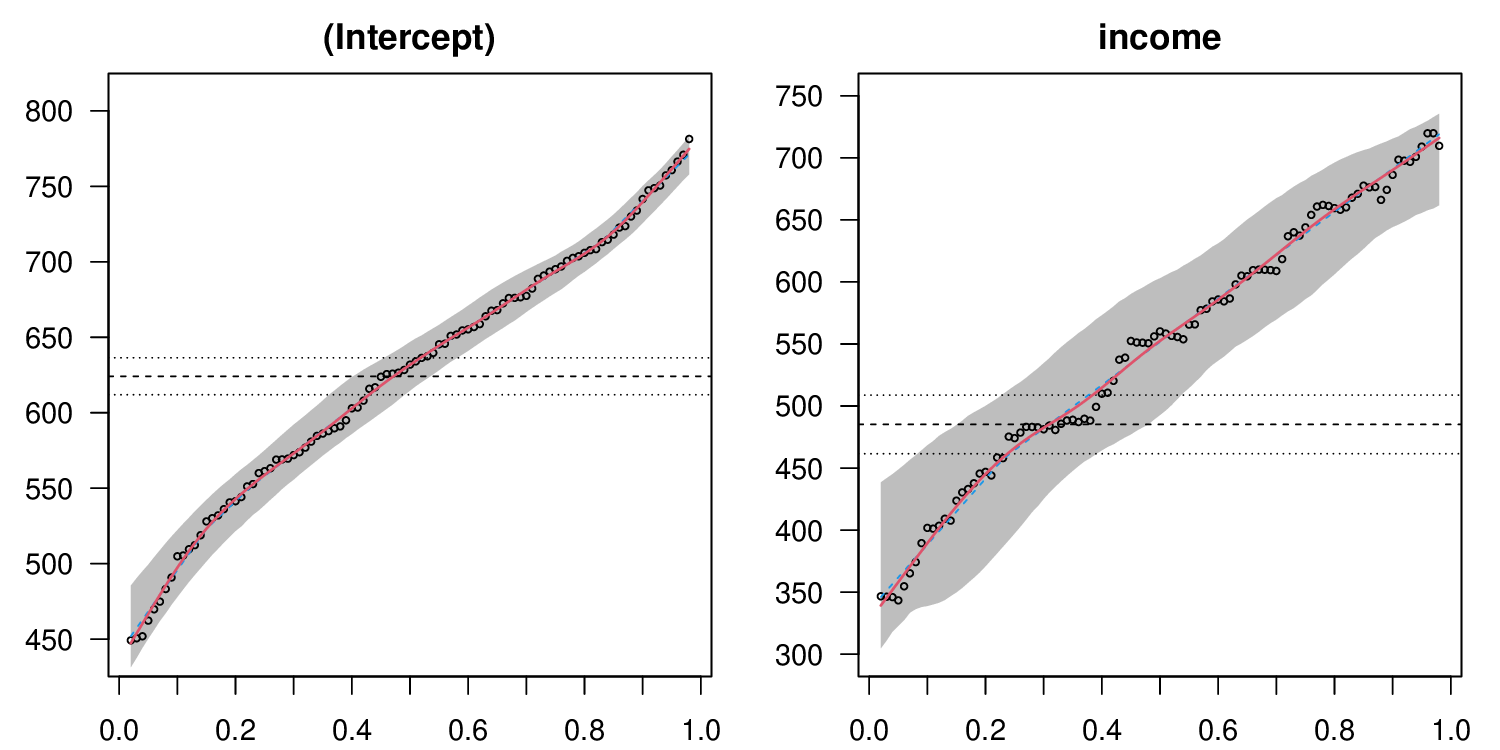} \\
{\small (b)}
\caption{Cubic SQR estimates of the intercept (left) and the coefficient of household income (right)  with a 90\% bootstrap confidence band for the Engel food expenditure data. The smoothing parameter is
selected by (a)  AIC and (b) BIC. 
Dashed line, linear SQR. Circles, QR.  Horizontal dashed lines and 
dotted lines show the ordinary least-squares estimate and a 90\% confidence interval in the style of Koenker (2005, p.\ 302). } 
\label{fig:engel}
\end{figure}

\begin{figure}[p]
\centering
\includegraphics[height=5.5in,angle=-90]{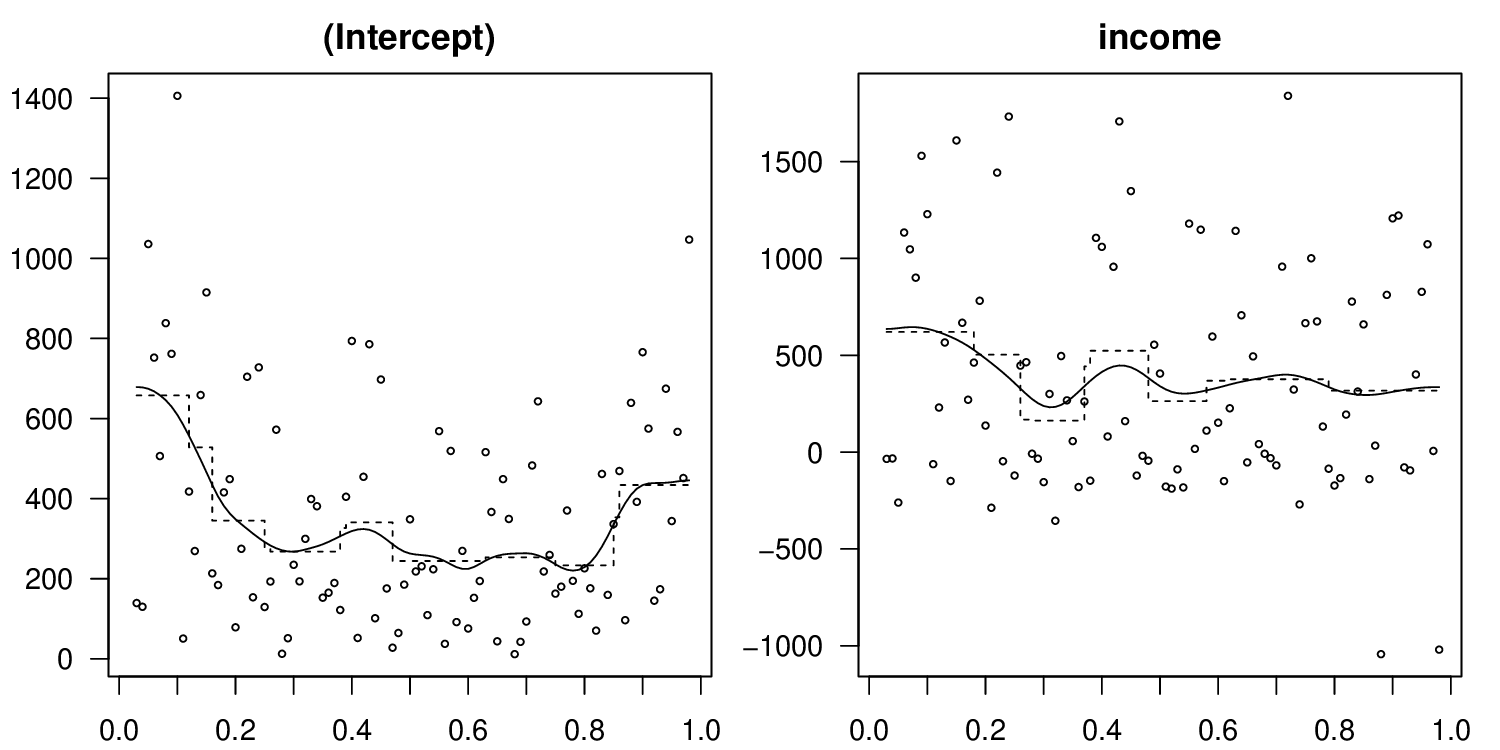} \\
{\small (a)} \\
\vspace{-0.2in}
\includegraphics[height=5.5in,angle=-90]{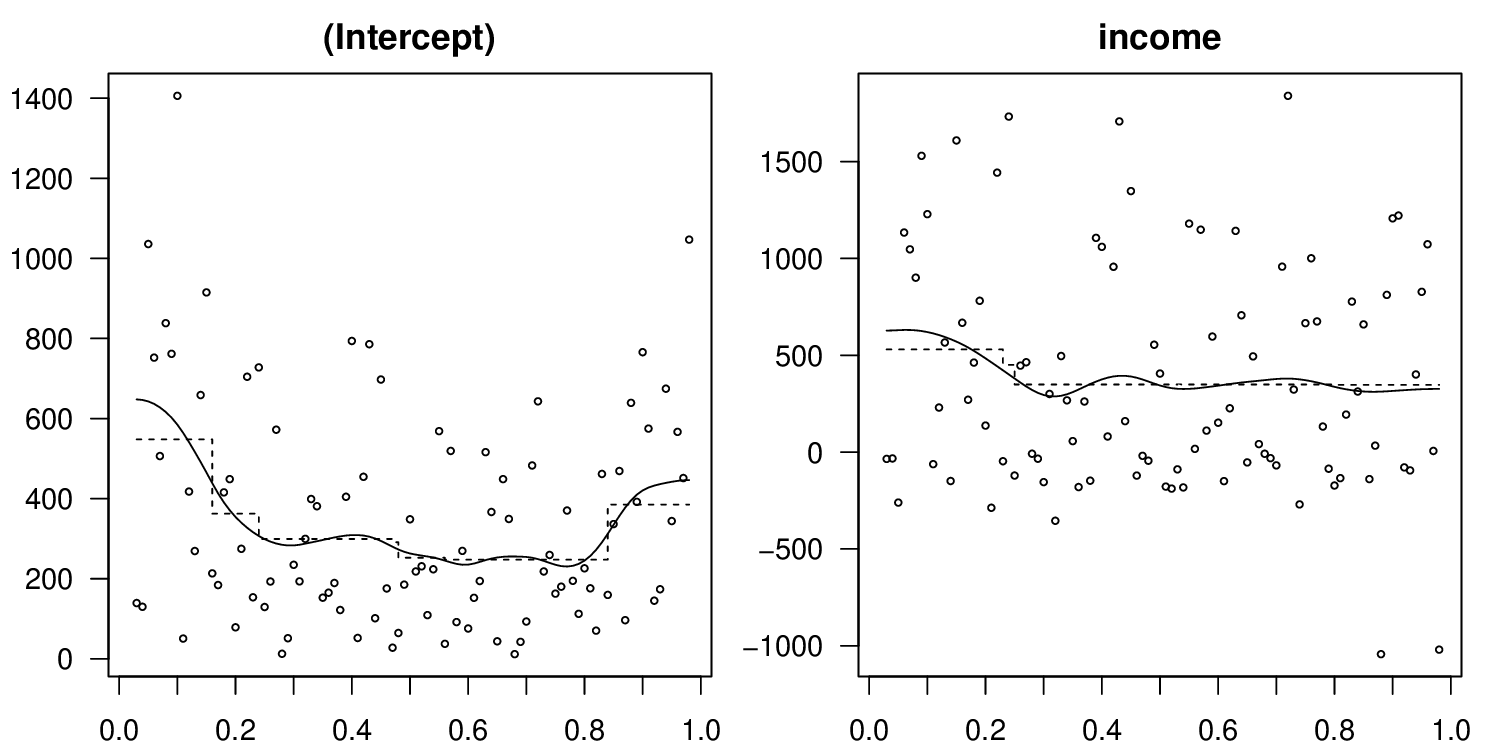} \\
{\small (b)}
\caption{Derivatives of the respective SQR estimates shown in Figure~\ref{fig:engel}.
Solid line, cubic SQR. Dashed line, linear SQR.  
Circles represent finite differences of the QR estimates. } 
\label{fig:engel2}
\end{figure}

The first real-data example is the Engel food expenditure data in the `quantreg' package (Koenker 2005, p.\ 78
and pp.\ 300--302). This data set contains $n=235$ records. The objective is to estimate the impact of the household income $x_t$ on the household food expenditure $y_t$ using a quantile regression model of the form (\ref{F}) with $\bx_t := [1, (x_t - \mu)/1000]^T$,  where   $\mu$ is the mean household income. 

We estimate  the intercept and the coefficient of household income
 as continuous functions of $\tau$ by solving the SQR problem in (\ref{sqr}) 
 with $\{ \tau_\ell \} :=  \{ 0.02,0.03,\dots,0.98\}$. This set of quantile levels was used in Koenker (2005, p.\ 300) to calculate the QR estimates. Figure~\ref{fig:engel} depicts the cubic SQR estimates (solid line) as functions of $\tau$
together with the linear SQR estimates (dashed line) and the QR estimates (circles).
In comparison with QR, the smoothing effect of SQR is particularly noticeable in the coefficient 
of household income. Moreover,  Figure~\ref{fig:engel}  includes
a 90\% pointwise bootstrap confidence band (shaded gray area) for the cubic SQR estimates.
It replaces the approximate confidence band used in Li and Megiddo (2026) which is based on the asymptotic
theory of the QR estimates and therefore does not account for the smoothing effect of SQR.
 Similarly to Figure A.2 in Koenker (2005, p.\ 302), Figure~\ref{fig:engel} also shows the estimated coefficients of least-squares regression  with a 90\% confidence  interval (dashed and dotted lines) for the conditional mean.
One can compare the least-squares estimates with the SQR estimates to assess the quantile-dependent patterns of the latter. 
 
The estimated coefficients by cubic SQR and linear SQR in Figure~\ref{fig:engel} are continuous functions of $\tau$.
 A difference between the cubic SQR estimates and the linear SQR estimates can be appreciated by inspecting 
 their derivatives shown in  Figure~\ref{fig:engel2}. According to  the analytical properties of cubic and linear splines,
the cubic SQR estimates should have continuous derivatives and  the linear SQR estimates
 should have piecewise-constant derivatives with the possibility of discontinuity or jump at some points in $\{ \tau_\ell \}$. This distinction in smoothness is clearly revealed by  Figure~\ref{fig:engel2}. 
Moreover,  Figure~\ref{fig:engel2} also demonstrates that the finite differences of the QR estimates (circles) are
 too noisy  for assessment of the changes in the coefficients with respect to $\tau$.

\medskip

\medskip

\begin{figure}[t]
\centering
\includegraphics[height=4in,angle=-90]{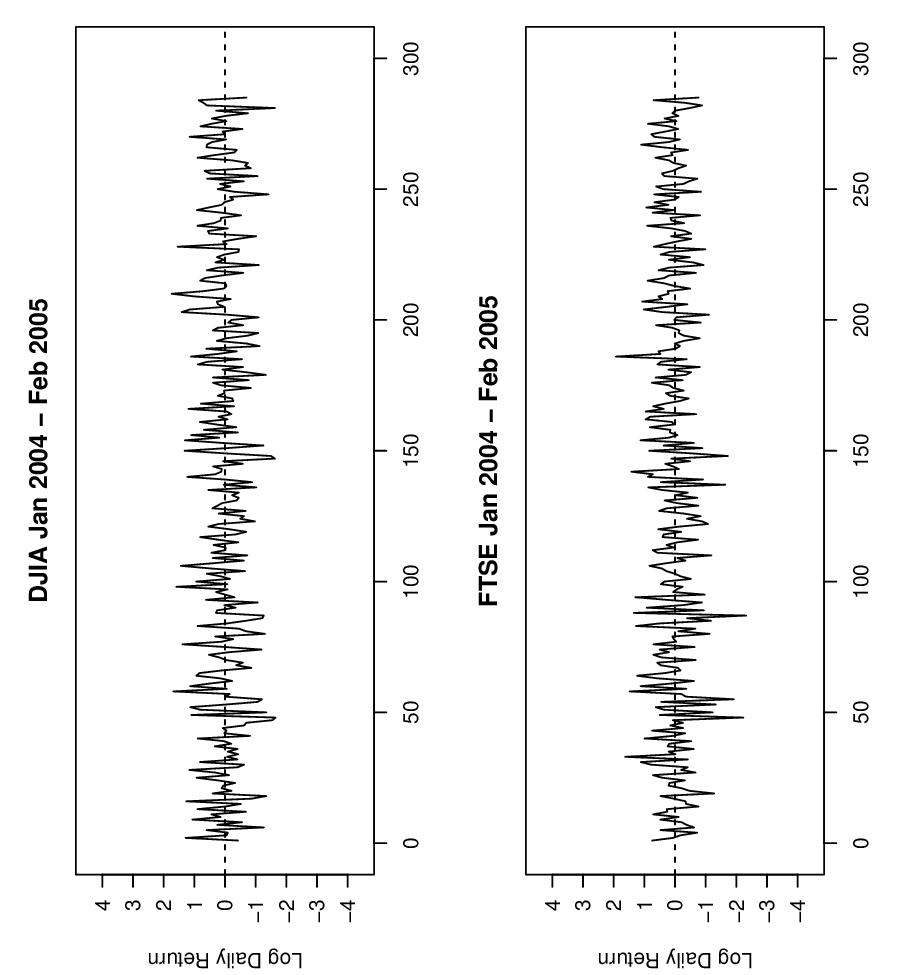} 
\caption{Log daily returns of DIJA (top) and FTSE (bottom) from January 2004 to February 2005. } 
\label{fig:fin1}
\end{figure}

The second real-data example concerns the log daily returns, shown in Figure~\ref{fig:fin1}, of the Dow Jones Industrial Average index  (DJIA) and the Financial Times Stock Exchange  100 index (FTSE).
As representations for stock market activities in New York and London, respectively, 
these indices have been used to study the relationship among international financial markets.
The concept of Granger causality (Granger 1969) is employed in particular to investigate 
the influence of one market on another.
Conventional methods  focus on the effect of Granger causality in the conditional mean using least-squares regression. More recent analyses use the quantile regression method to study the effect of Granger causality  in the conditional quantiles (Cheng et al.\ 2022; Troster 2022; Karpman et al.\ 2022; Hu\'{e} and Leymarie 2025).
These analyses are typically confined to a few selected quantiles. The SQR method 
discussed in this article enables a more comprehensive analysis across all quantiles inside an interval $[a,b] \subset (0,1)$.

With $v_t$ denoting the closing value of FSTE on trading day $t$, the log daily return of FSTE on day $t$ 
is defined as $y_{t} := \log(v_{t}/v_{t-1})$. Similarly, let $x_t$ denote the log daily return of DJIA on day $t$.
Then, the Granger causality of DJIA on FTSE, denoted by DJIA$\rightarrow$FTSE,  can be analyzed through 
a QAR model (Koenker and Xiao 2006) of the form
\eqn
y_t = a_0(u_t) + a_1(u_t) y_{t-1} + c(u_t) x_{t-1},
\label{gc}
\eqqn
where $\{ u_t\}$ is an i.i.d.\ sequence of $U(0,1)$ random variables. The 
conditional quantile function of $y_t$ given $y_{t-1}$ and $x_{t-1}$ takes the form (\ref{F}) with $\bx_t := [1,y_{t-1},x_{t-1}]^T$ and $\bmbeta_0(\cdot) := [a_0(\cdot),a_1(\cdot),c(\cdot)]^T$. For each $\tau$, 
the quantity $c(\tau)$ is the coefficient of feedback from DJIA 
on day $t-1$ to the $\tau$th conditional quantile of FTSE on day $t$, subtracting 
the feedback from FTSE on day $t-1$.
. If $c(\tau) \ne 0$, then the Granger causality  DJIA$\rightarrow$FTSE exists at the $\tau$th conditional quantile.
In this case, one can use $|c(\tau)|$ to measure the strength of the causality. A positive causality, $c(\tau) > 0$, implies that  the $\tau$th quantile of FTSE on day $t$ moves  in the same direction as DJIA on day $t-1$, whereas a negative causality, $c(\tau) < 0$, implies that
the quantile moves in the opposite direction of DJIA. 
Similarly, the Granger causality of FTSE on DJIA, denoted by FTSE$\rightarrow$DJIA, can be analyzed using the QAR model (\ref{gc}) except that $\{ y_t \}$ stands for DJIA and $\{ x_t\}$  for FTSE.

\begin{figure}[p]
\centering
\includegraphics[height=5.5in,angle=-90]
{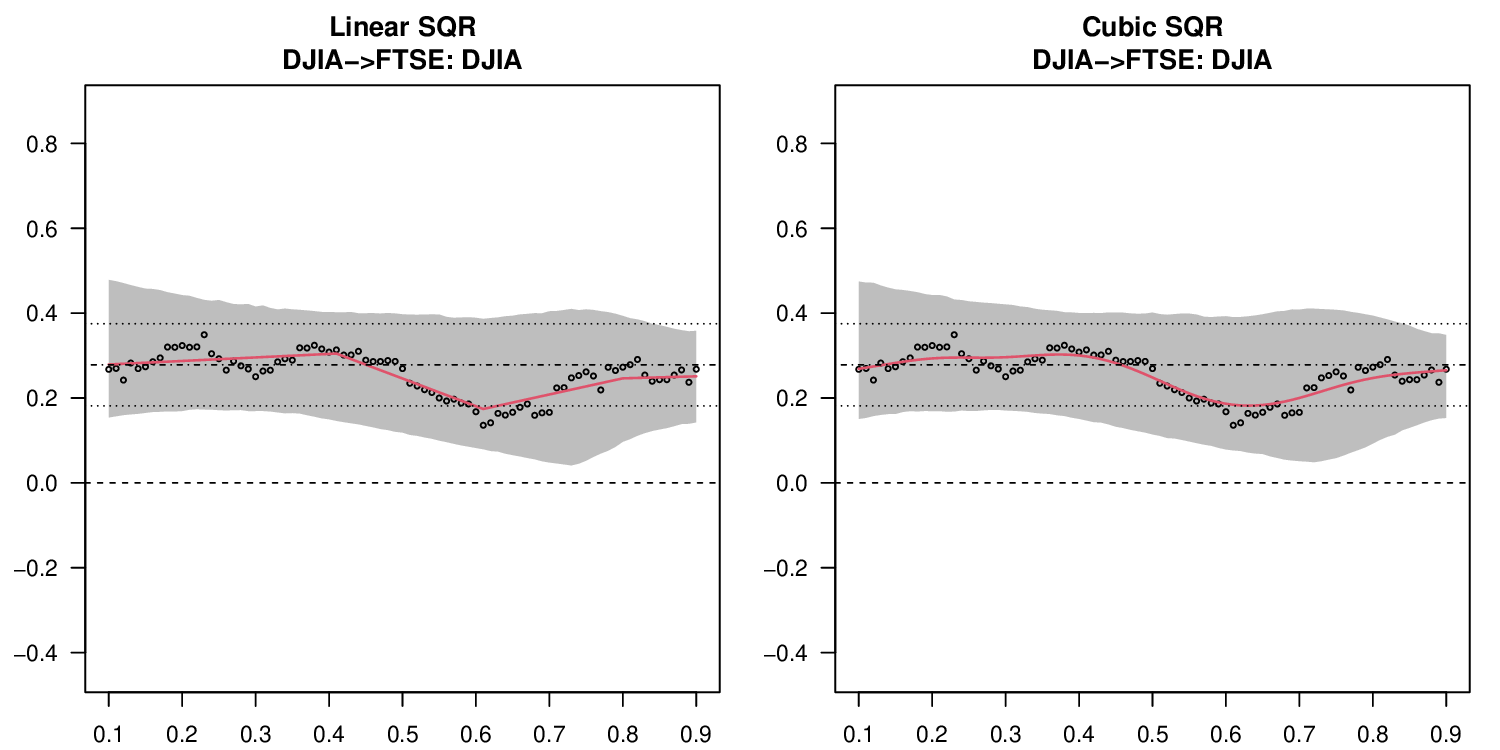} \\
{\small (a)} \\
\vspace{-0.2in}
\includegraphics[height=5.5in,angle=-90]
{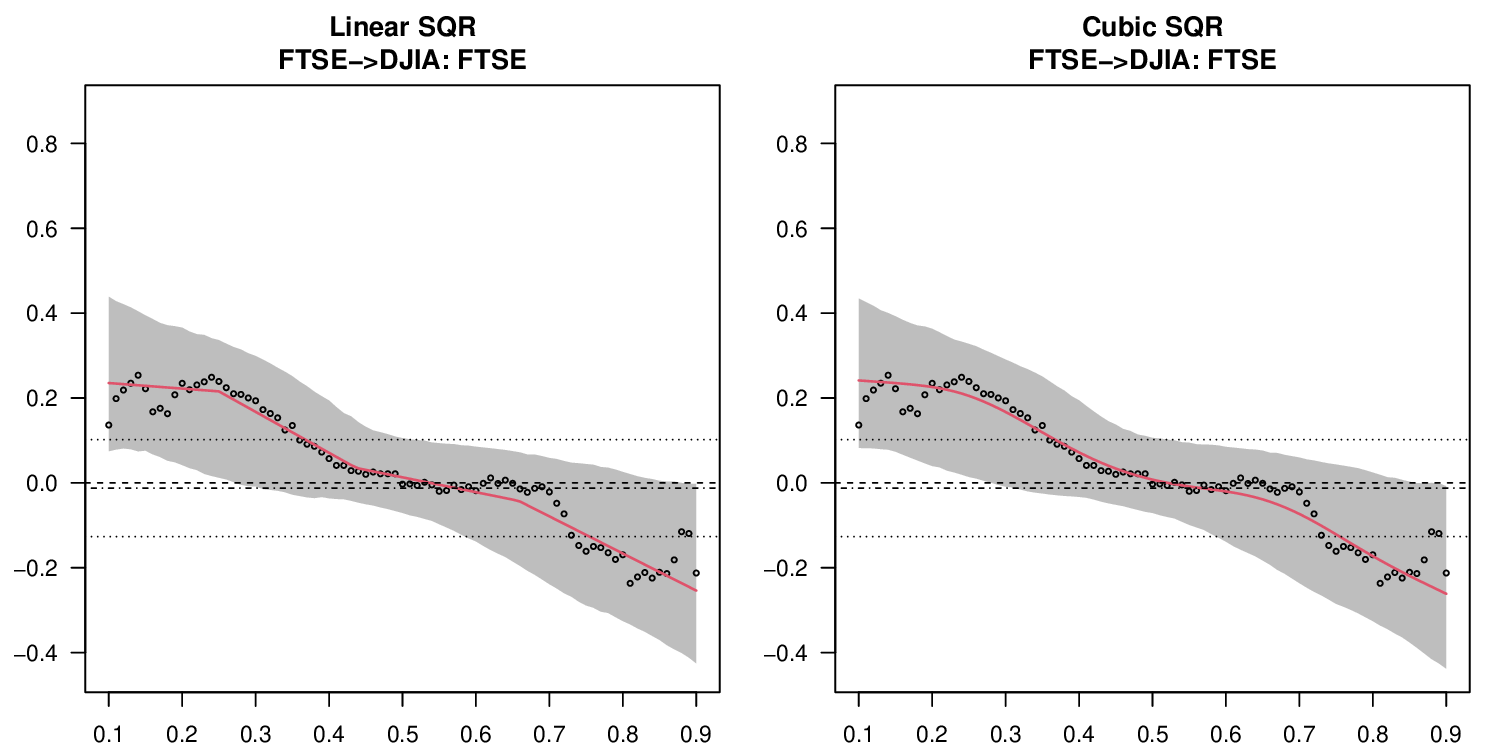} \\
{\small (b)}
\caption{Analysis of Granger causality in quantiles for the log daily returns of DJIA and FTSE from July 2007 to August 2008 shown in Figure~\ref{fig:fin1}. (a) Linear and cubic SQR estimates of $c(\cdot)$  for DJIA$\rightarrow$FTSE.
(b) Linear and cubic SQR estimates of $c(\cdot)$ for FTSE$\rightarrow$DJIA. Smoothing parameters are selected by AIC. Gray area depicts a 90\% bootstrap confidence band by block sampling with block-length equal to 10. Circles represent QR estimates. Horizontal dash-dotted line and dotted lines show 
the ordinary least-squares estimate and a 90\% confidence interval.} 
\label{fig:fin1b}
\end{figure}

To illustrate the implication of Granger causality in quantiles, consider a  hypothetical case of DJIA$\rightarrow$FTSE with $c(0.9) > 0$ . Suppose  the 0.9th quantile of the return of FTSE equals 0.6 when the return of DJIA is zero  on the previous day. This quantile will move up by the amount of $c(0.9)$ if the return of DJIA on the previous day becomes 1 and the return of FTSE on the previous day remains unchanged. An increased 0.9th quantile implies a higher than 0.1 probability for the return of FSTE to take values greater than 0.6. Such effect  may be summarized by simply saying that a positive Granger causality of DJIA on FTSE at higher quantiles would increase the chance for the return of FTSE to take larger values when DJIA rallies on the previous day. Similarly, a positive Granger causality of DJIA on FTSE at lower quantiles would increase the chance for the return of  FTSE to take smaller or more negative values when DJIA declines on the previous day. The amount of increased chance in both cases is directly related to the magnitude of $c(\cdot)$.

Based on the data shown  in Figure~\ref{fig:fin1},
 the linear and cubic SQR estimates of $c(\cdot)$ are depicted in Figure~\ref{fig:fin1b} for detecting the Granger causality DJIA$\rightarrow$FTSE and FTSE$\rightarrow$DJIA, respectively, under the QAR framework (\ref{gc}).
Figure~\ref{fig:fin1b}  also shows the QR estimates  for comparison. 
In addition, Table~\ref{tab:fin1} contains the numerical values of the SQR estimates  at selected quantiles $\tau=0.1$, 0.5, and 0.9.
The confidence bands in Figure~\ref{fig:fin1b} and the confidence intervals in Table~\ref{tab:fin1} 
are constructed by the block bootstrap method. A comparison with the confidence bands 
constructed by the ordinary $(x,y)$-pair sampling 
scheme can be found in the supplementary material.

\begin{table}[t]
\begin{center} 
\caption{Estimation of Granger Causality at Selected Quantiles: Jan 2004 - Feb 2005} \label{tab:fin1}
{\small
\begin{tabular}{c|cc|cc} \hline
& \multicolumn{2}{c|}{DJIA$\rightarrow$FTSE} & \multicolumn{2}{c}{FTSE$\rightarrow$DJIA} \\
 $\tau$ & Linear SQR & Cubic SQR &  Linear SQR & Cubic SQR  \\  \hline
 0.1  & 0.279 (0.154,0.478) & 0.269 (0.151,0.474) & 0.235 (0.075,0.439)   & 0.241 (0.082,0.434)   \\
 0.5  & 0.246 (0.119,0.397) & 0.248 (0.114,0.401) & 0.013 (-0.071,0.105)  & 0.008 (-0.071,0.106) \\
 0.9  & 0.251 (0.143,0.358) & 0.265 (0.153,0.348) & -0.254 (-0.425,-0.003) & -0.261 (-0.435,-0.008) \\
\hline
\end{tabular} 
}
\end{center}
{\footnotesize 
\begin{center}
\begin{minipage}{6in}
Parentheses contain the lower and upper limits of the 90\% confidence band.
\end{minipage}
\end{center}
}
\end{table}

Figure~\ref{fig:fin1b}(a) suggests a positive Granger causality DJIA$\rightarrow$FTSE 
across all quantiles because the confidence band is  well above the zero line which
 represents the case of  no Granger causality $c(\cdot) = 0$. 
 Such causality  is also detected in the conditional  mean by least-squares regression
 (horizontal dash-dotted and dotted lines). Figure~\ref{fig:fin1b}(a) shows that despite a small dip in the upper middle  quantile region, both  linear and cubic SQR estimates of $c(\cdot)$ for  DJIA$\rightarrow$FTSE are largely leveled around the average value  0.254 or 0.255, respectively. 
 In contrast, the SQR estimates for  FSTE$\rightarrow$DJIA, shown in Figure~\ref{fig:fin1b}(b), exhibit a notable decreasing trend which goes from 0.235 or 0.241  at $\tau=0.1$ to $-0.254$ or $-0.261$ at $\tau=0.9$ according to Table~\ref{tab:fin1}.  This leads to a positive causality at low quantiles ($\tau \le 0.27$ or 0.28)
and a possible negative causality  at high quantiles ($\tau \ge 0.88$ or 0.86). 
This quantile-dependent pattern of causality is not reflected in the conditional mean by least-squares regression.

\begin{figure}[t]
\centering
\includegraphics[height=4in,angle=-90]{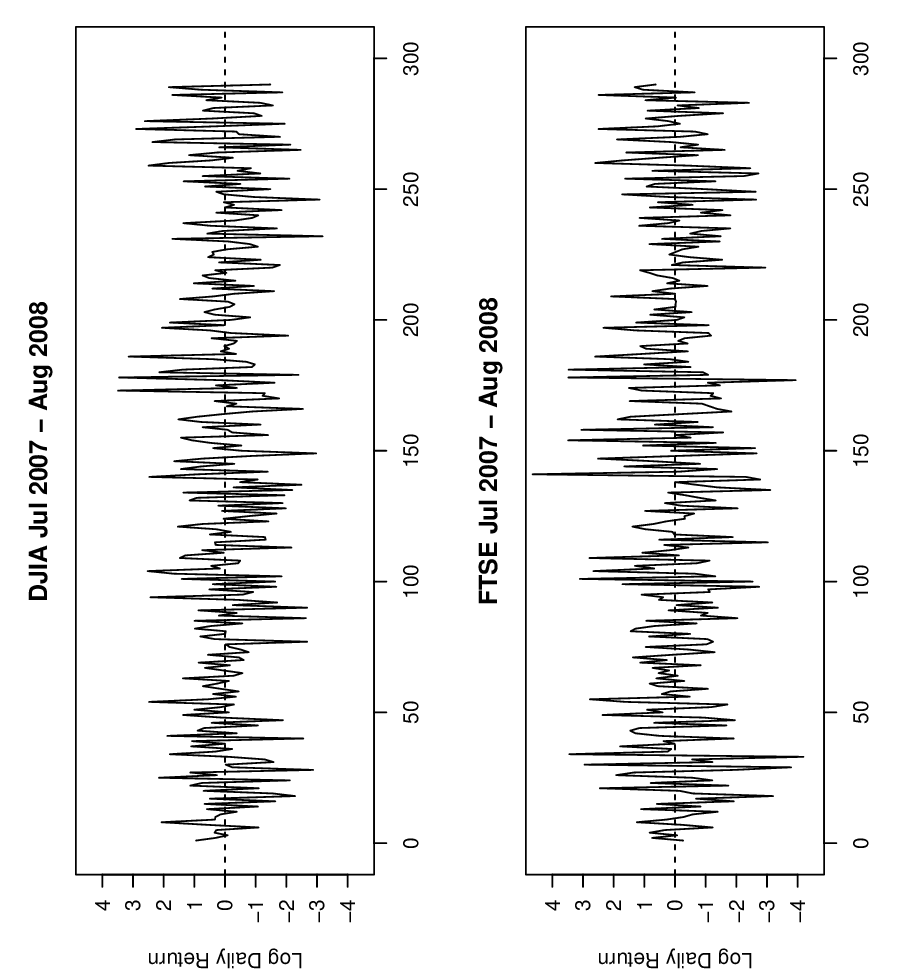} 
\caption{Log daily returns of DIJA (top) and FTSE (bottom) from July 2007 to August 2008. } 
\label{fig:fin2}
\end{figure}

The 13-month period  from January 2004 to February 2005 is apparently less volatile than 
the period from July 2007 to August 2008 for which  the log daily returns are shown  in Figure~\ref{fig:fin2}. 
The latter is a period that leads to the stock market crash caused by the collapse of Lehmann Brothers. 
An interesting question is whether the increased volatility coexists with a change in the patterns of Granger causality 
between DJIA and FTSE. An affirmative answer is supported by the result  in Figure~\ref{fig:fin2b}, which depicts 
the SQR estimates of $c(\cdot)$ for DJIA$\rightarrow$FTSE and FTSE$\rightarrow$DJIA, respectively, 
based on the data shown in Figure~\ref{fig:fin2}. It is also supported by the result in  Table~\ref{tab:fin2}, which contains the numerical values of the SQR estimates  at selected quantiles.

\begin{figure}[p]
\centering
\includegraphics[height=5.5in,angle=-90]
{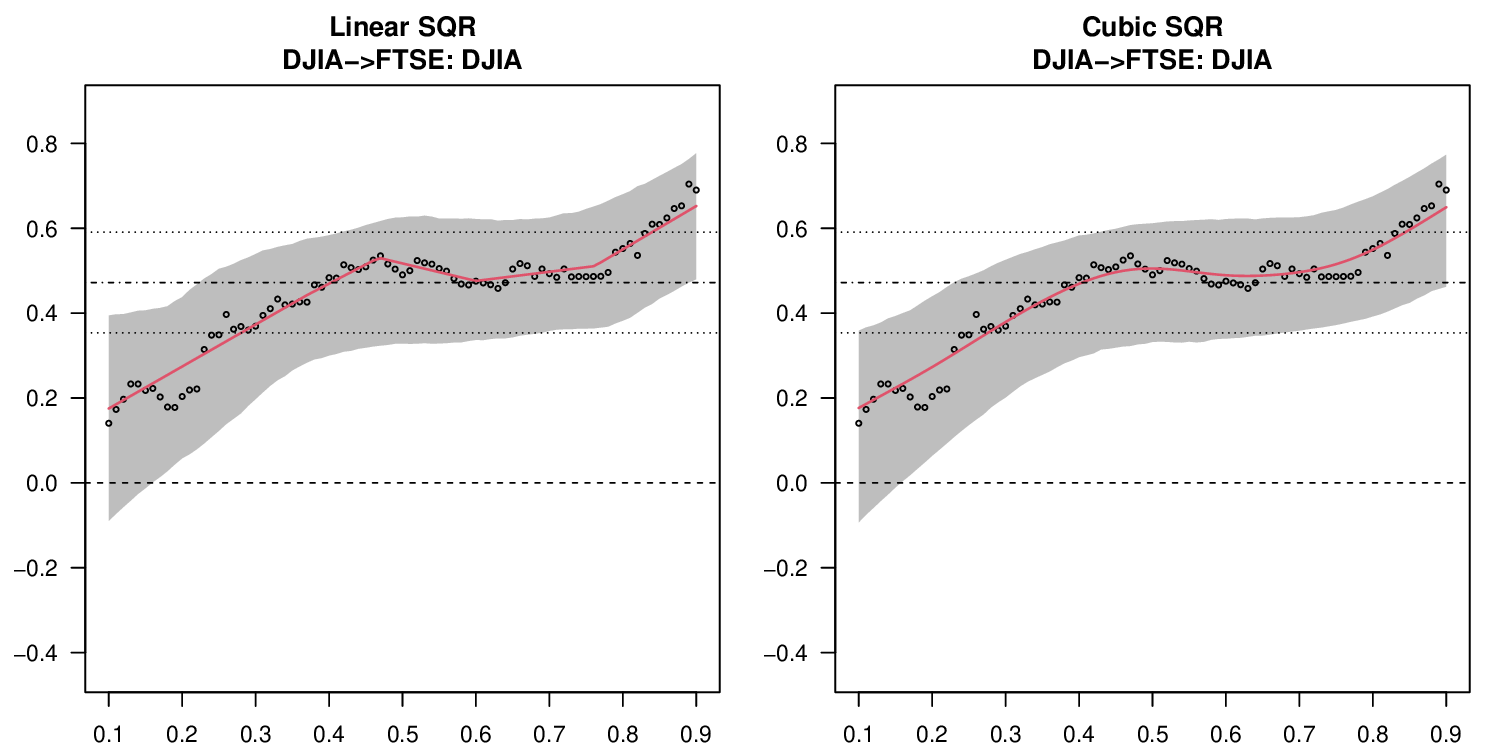} \\
{\small (a)} \\
\vspace{-0.2in}
\includegraphics[height=5.5in,angle=-90]
{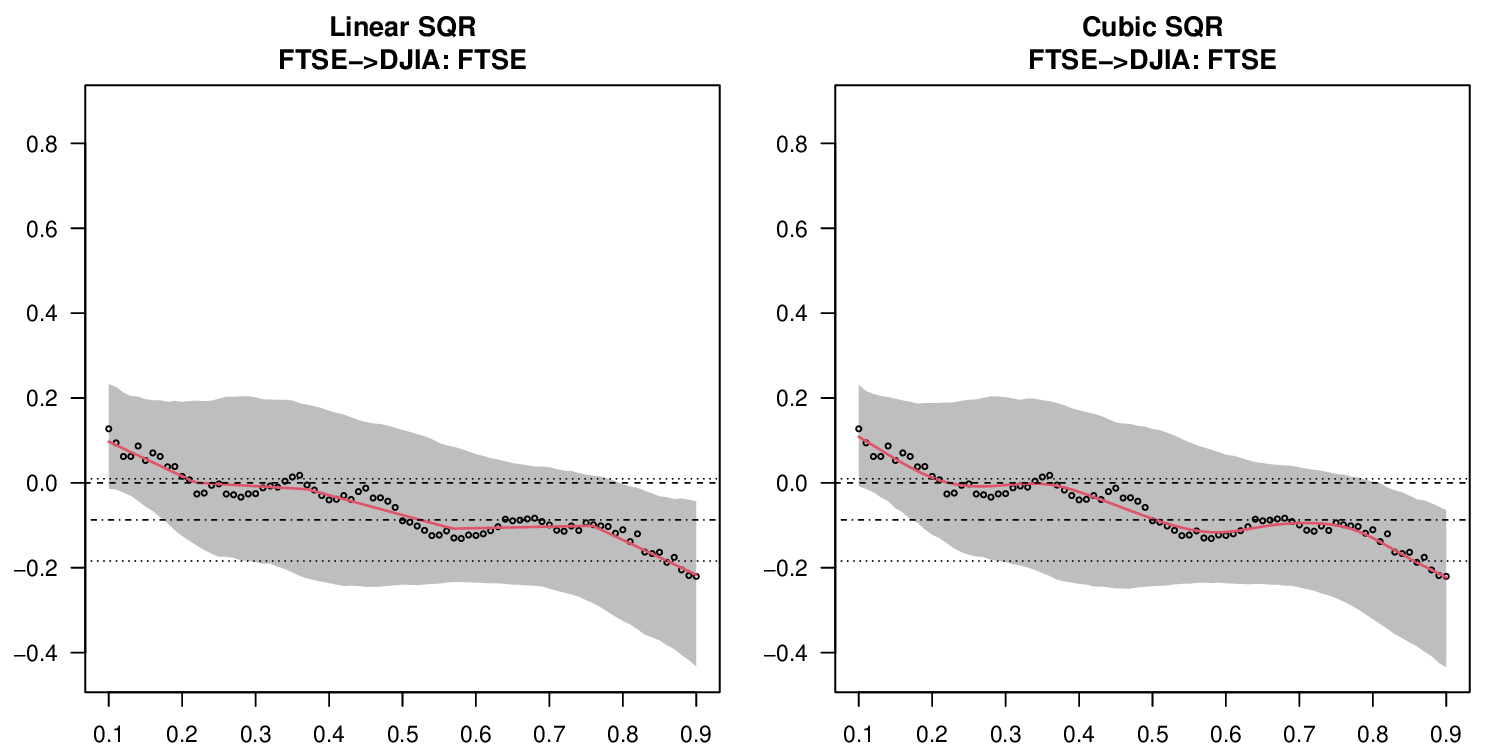} \\
{\small (b)}
\caption{Same analysis as Figure~\ref{fig:fin1b} for the log daily returns of DJIA and FTSE from July 2007 to August 2008 shown in Figure~\ref{fig:fin2}. } 
\label{fig:fin2b}
\end{figure}

\begin{figure}[p]
\begin{center} 
\caption{Estimation of Granger Causality at Selected Quantiles: Jul 2007 - Aug 2008} \label{tab:fin2}
{\small
\begin{tabular}{c|cc|cc} \hline
& \multicolumn{2}{c|}{DJIA$\rightarrow$FTSE} & \multicolumn{2}{c}{FTSE$\rightarrow$DJIA} \\
 $\tau$ & Linear SQR & Cubic SQR &  Linear SQR & Cubic SQR  \\  \hline
 0.1  & 0.175 (-0.089,0.394) & 0.176 (-0.093,0.359) & 0.097 (-0.014,0.233) & 0.109 (-0.007,0.231)   \\
 0.5  & 0.517 (0.328,0.626)  & 0.505 (0.332,0.612)  & -0.076 (-0.240,0.124) & -0.085 (-0.243,0.127)  \\
 0.9  & 0.652 (0.480,0.777)  & 0.649 (0.462,0.773)  & -0.217 (-0.432,-0.044) & -0.223 (-0.434,-0.065) \\
\hline
\end{tabular} 
}
\end{center}
{\footnotesize 
\begin{center}
\begin{minipage}{6in}
Parentheses contain the lower and upper limits of the 90\% confidence band.
\end{minipage}
\end{center}
}
\centering
\vspace{0.1in}
\includegraphics[height=2.75in,angle=-90]
{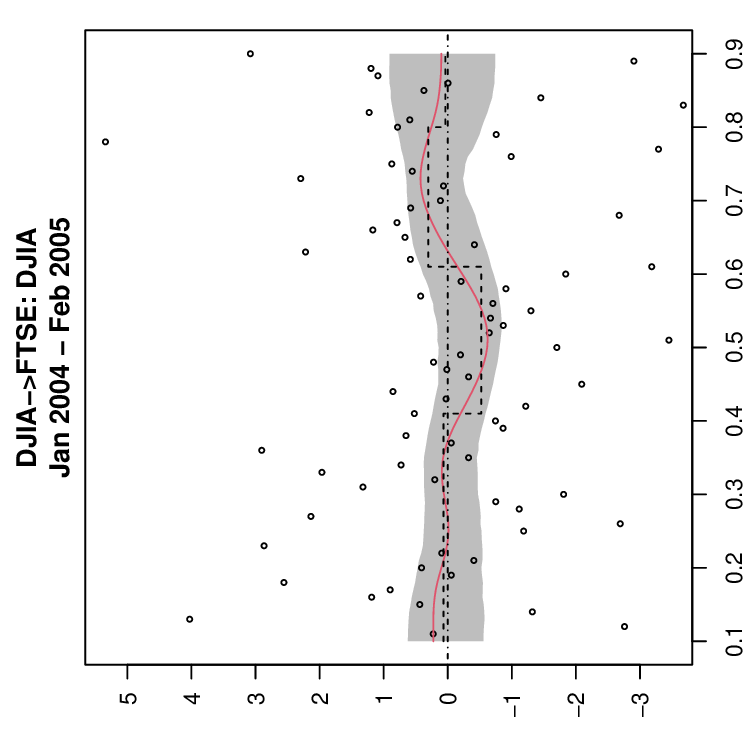}
\includegraphics[height=2.75in,angle=-90]
{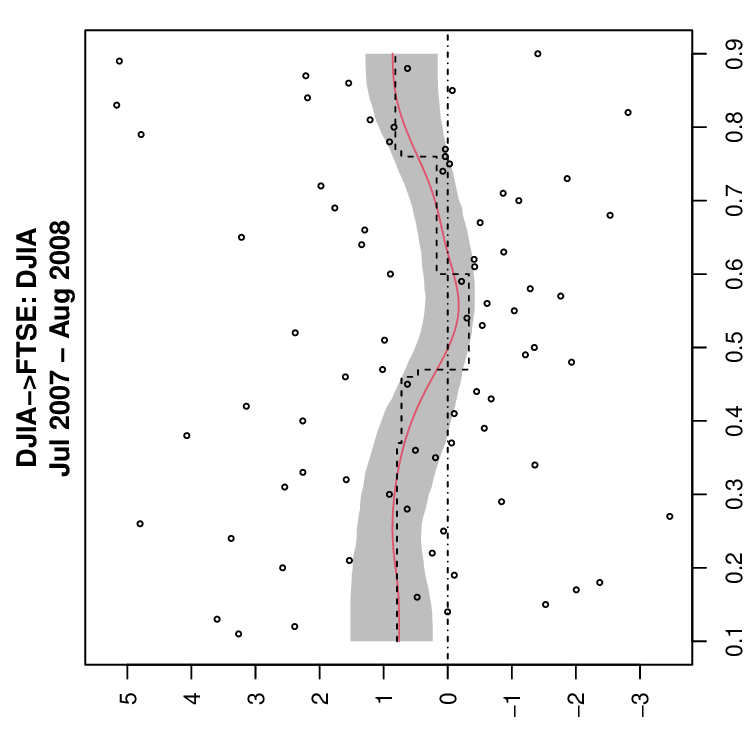} 
\caption{Derivative of cubic SQR estimate of $c(\cdot)$ with a 90\% bootstrap confidence band for DJIA$\rightarrow$FTSE from January 2004 to February 2005 (left) and July 2007 to August 2008 (right). 
Dashed line, derivative of linear SQR estimate; circles, finite difference of QR estimate. } 
\label{fig:fin2c}
\end{figure}

The increase of the least-squares estimate (dashed line) in Figure~\ref{fig:fin2b}(a)
in comparison with  Figure~\ref{fig:fin1b}(a) suggests an overall  strengthening 
of the Granger causality DJIA$\rightarrow$FTSE in the second period.  
However,  according to  the SQR estimates in Figure~\ref{fig:fin2b}(a), 
this strengthening is not distributed equally across all  quantiles.
Indeed,  the magnitude of $c(\tau)$  grows with the increase of $\tau$ at high quantiles ($\tau > 0.8$) but falls 
with the decrease of $\tau$ at low quantiles $(\tau < 0.4$). This is in contrast with the largely flat pattern in Figure~\ref{fig:fin1b}(a). A comparison of their derivatives  shown in Figure~\ref{fig:fin2c} further confirms 
this difference between the two periods.
Table~\ref{tab:fin2} shows that the SQR estimates are 0.652 and 0.649 at $\tau=0.9$ versus 0.175 and 0.176  at $\tau=0.1$. Such large disparity is not exhibited  in  Table~\ref{tab:fin1}. 
These results suggest that DJIA in the second period becomes more influential on big rallies of FTSE relative to 
its influence  on big declines  of FTSE.

With regard to the Granger causality FTSE$\rightarrow$DJIA,  one can see, by comparing Figure~\ref{fig:fin2b}(b) and Table~\ref{tab:fin2} with Figure~\ref{fig:fin1b}(b) and Table~\ref{tab:fin1}, that the decreasing pattern of the SQR estimates remains in the second period, but the assertion of possible causality at low quantiles is no longer supported by the SQR estimates. as judged by the confidence limits.

\medskip

Finally, we note that the supplementary material contains additional examples of application. In the following
we only highlight some computational issues encountered in the analysis of the birth data. 
This example investigates the impact of various demographic characteristics and maternal behaviors on the birth weight of infants in the US. Similar data were used by Koenker (2005, p.~20) and by Li and Megiddo (2026).
 The large number of data records employed in this example ($n=20000$) poses some challenges to the computation of the SQR solutions discussed in this article. For linear SQR, 
we find it beneficial to use the LP solver {\tt rq.fit.sfn} instead of {\tt rq.fit.fnb2} for reduced computer time. For cubic SQR, the QP solver {\tt solve\underline{ }piqp}  crashes without an error message. The QP solver {\tt solve\underline{ }osqp} is able to compute the cubic SQR estimates, but the computer time is much longer than that for computing the linear SQR estimates, which makes the construction of bootstrap confidence bands especially time-consuming.
In addition, we find it helpful to rescale the variables so that their coefficients in the SQR problem have roughly the same dynamic range across the quantiles suitable for regularization with a single smoothing parameter.  More details  can be found in the supplementary material.

\section{Concluding Remarks}

This article considers the problem of estimating the coefficients in linear quantile regression models as smooth functions of the quantile level by the spline quantile regression (SQR) method of Li and Megiddo (2026).
The SQR method represents the coefficients by spline functions and produces an estimate of the functional coefficients  by solving a penalized quantile regression problem  with penalty on their roughness at a given set of quantiles.. This article extends the SQR method  by
introducing  additional choices for the  functional space and  the roughness penalty. The combination of the $\ell_2$-norm of the second derivatives  as the penalty with the space of cubic splines yields the cubic SQR problem that can be formulated as a quadratic program (QP). The combination of the total variation of the first derivatives as the penalty  with the space of linear splines leads to  the linear SQR problem that  can be solved as a linear program (LP). Both cubic and linear SQR solutions are shown to be optimal in a space of smooth functions beyond their respective space of cubic or linear splines with fixed knots at the given set of quantiles. 

The SQR method offers a concise representation of the regression coefficients across all quantiles. 
It enables an assessment of the derivatives of the regression coefficients as functions of the quantile level.
Our simulation study demonstrates that the SQR method also has the capability of producing more 
accurate estimates in comparison with the conventional quantile regression method performed 
independently at each quantile even  with ad hoc post-smoothing.  

It remains desirable to devise special algorithms for the cubic SQR problem by taking full advantage of the 
special structure of its QP formulation for accelerated convergence without excessive consumption of computer memory. Alternative complexity measures that vary smoothly with the smoothing parameter in the SQR problem are also desirable. Rescaling the explanatory variables to equalize the dynamic range of the regression coefficients across quantiles helps improve the effectiveness of regulating their smoothness  using a single smoothing parameter. 
Extension of the SQR method to include multiple smoothing parameters may enhance its flexibility when  dealing with complex situations where different  coefficients demand different degrees of smoothness. Additional lasso-type penalties on the spline parameters as in Yoshida (2021) may further enhance the capability of SQR when a sparse model is called for.

It is well known that the ordinary QR solutions  may sometimes violate 
the monotonicity of conditional quantiles (Koenker 2005, p.\ 55;  Neocleous and Portnoy 2008). 
Remedies to this quantile-crossing problem have been proposed under various
conditions (He 1997; Wu and Liu 2009; Bondell et al.\ 2010; Ando and Li 2025; Szendrei et al.\ 2025). 
For the SQR estimators given by (\ref{sqr}), it is conceivable,  as noted in Li and Megiddo (2026), 
that a constraint of the form $\bx_t^T \hat{\bmbeta}(\tau_{\ell+1}) \ge \bx_t^T \hat{\bmbeta}(\tau_{\ell})$
$(\ell=1,\dots,L-1;t=1,\dots,n)$ could be imposed 
without altering the respective QP or LP nature of the problem, except for an increase in the computational burden. 
This constraint guarantees the monotonicity of the solution at the observed design points 
and the finite set of quantiles. Extension of this monotonicity  beyond the observed 
design points and for all quantiles is an interesting problem for future research.

\section*{References}

{\footnotesize
\begin{description} 

\item
Ando, T., and Li, K.-C. (2025). Simplex quantile regression without crossing.
{\it Annals of Statistics}, 53, 144--169.

\item
Andriyana, Y., Gijbels, I., and Verhasselt, A. (2014). P-splines quantile regression estimation in varying
coefficient models. {\it Test}, 23, 153--194. 

\item
Bondell, H., Reich, B., and Wang, H. (2010).
Noncrossing quantile regression curve estimation.
{\it Biometrika}, 97, 825--838.

\item
Brezis, H. (2019). Regularized interpolation driven by total variation. {\it Analysis in Theory and Applications},
35, 335--354.

\item
B\"{u}hlmann, P. (2002). Bootstrap for times series, {\it Statistical Science}, 17, 52--72.

\item
Cheng, H., Wang, Y., Wang, Y., and Yang, T. (2022). Inferring causal interactions in financial markets using conditional Granger causality based on quantile regression, {\it Computational Economics}, 59, 719--748.

\item
 Chernozhukov, V., Fern\'{a}ndez-Val, I., and Melly B. (2022).
 Fast algorithms for the quantile regression process. {\it Empirical Economics}, 62, 7--33.

\item
Clette, F., Svalgaard, L., Vaquero, J., and Cliver, E. (2014). Revisiting the sunspot number. A 400-year perspective on the solar cycle. {\it Space Science Reviews}, 186, 35--103. 

\item
de Boor, C. (1963). Best approximation properties of spline functions of odd degree,
{\it Journal of Mathematics and Mechanics}, 12, 747--749.

\item
Dorn, W. (1960). Duality in quadratic programming. {\it Quarterly of Applied Mathematics}, 18, 155--162.

\item
Friedlander, M., and Orban, D. (2012).  A primal-dual regularized interior-point method for 
convex quadratic programs. {\it Mathematical Programming Computation}, 4, 71--107.

\item
Granger, C. (1969). Causal relations by econometric models and cross-spectral methods. {\it Econometrica}, 37, 424-438.

\item
Hastie, T., and Tibshirani, R. (1990). {\it Generalized Additive Models}, Chapter 2.
New York: Chapman \& Hall.

\item
He, X. (1997). Quantile curves without crossing. {\it American Statistician}, 51, 186--192.

\item
Hu\'{e}, S., and Leymarie, J. (2025). Granger-causality in quantiles and financial interconnectedness.
\url{https://dx.doi.org/10.2139/ssrn.5861723}

\item
Karpman, K., Lahiry, S., Mukherjee, D., and Basu, S. (2022).
Exploring financial networks using quantile regression and Granger causality,  	arXiv:2207.10705.

\item
Koenker, R. (2005). {\it Quantile Regression}. Cambridge University Press, Cambridge.

\item
Koenker, R., and Bassett, G. (1978). Regression quantiles. {\it Econometrica}, 46, 33--50.

\item
Koenker, R., and D'Orey, V. (1987). Algorithm as 229: Computing regression quantiles. 
{\it Journal of the Royal Statistical Society Series C}, 36, 383--393.

\item
Koenker, R., and Ng, P. (2005).
A Frisch-Newton algorithm for sparse quantile regression.
{\it Acta Mathematicae Applicatae Sinica},  21, 225--236.

\item
Koenker, R., Ng, P., and Portnoy, S. (1994). Quantile smoothing splines. {\it Biometrika}, 81, 673--680.

\item
Koenker, R., and Xiao, Z. (2006). Quantile autoregression, {\it Journal of the American Statistical Association}, 101, 980--990.

\item
Lahiri, S. (2003). {\it Resampling Methods for Dependent Data}. New York: Springer.

\item 
Li, T.-H., and Megiddo, N.  (2026). Spline quantile regression. {\it Journal of Statistical Theory and Practice}, 20, 30. 
\url{https://doi.org/10.1007/s42519-026-00545-8}.

\item
Narasimhan, B., Schwan, R., Jiang, Y., Goulart, P., Kuhn, D., and Jones, C. (2023). Package ‘piqp’ (Version 0.2.2).
\url{https://cran.r-project.org/web//packages/piqp/piqp.pdf}.

\item
Neocleous, T.,  and Portnoy, S.  (2008). On monotonicity of regression quantile functions. 
{\it Statistics \& Probability Letters},  78, 1226--1229.

\item
Oh, H.-S., Lee, T., and Nychka, D. (2011).
Fast nonparametric quantile regression with arbitrary smoothing methods.
{\it Journal of Computational and Graphical Statistics}, 20, 510--526.

\item
Park, S., and He, X. (2017). Hypothesis testing for regional quantiles. 
{\it Journal of Statistical Planning and Inference}, 191, 13--24.

\item 
Pinkus, A. (1988). On smoothest interpolants. {\it SIAM Journal of Mathematical Analysis}, 19, 1431--1441.

\item
Portnoy, S., and Koenker, R. (1997). The Gaussian hare and the Laplacian tortoise: computability of squared-error
versus absolute-error estimators. {\it Statistical Science}, 12, 279--300.

\item
R Core Team (2024). R: A language and environment for statistical
  computing. R Foundation for Statistical Computing, Vienna,
  Austria. \url{https://www.R-project.org/}.

\item
Schwan, R., Jiang, Y., Goulart, P., Kuhn, D., and Jones, C. (2023). PIQP: A proximal interior-point quadratic 
programming solver. arXiv:2304.00290.

\item
Stellato, B., Banjac, G., Goulart, P., Bemporad, A., and Boyd, S. (2020). 
OSQP: An operator splitting solver for quadratic programs. {\it Mathematical Programming Computation}, 
12, 637--672. 

\item
Stellato, B., Banjac, G., Goulart, P., Boyd, S., Anderson, E., Bansal, V., and Narasimhan, B. 
(2024). Package `osqp' (Version 0.6.3.3). \url{https://cran.r-project.org/web/packages/osqp/osqp.pdf}.

\item
Stoer, J., and Bulirsch, R. (2002). {\it Introduction to Numerical Analysis}, 3rd Edition, Chapter 2.
New York: Springer.

\item
Szendrei, T., Bhattacharjee, A., and Schaffer, M. (2025). Fused LASSO as non-crossing quantile
regression. arXiv:2403.14036.

\item
Troster, V. (2018). Testing for Granger-causality in quantiles.
{\it Econometric Reviews}, 37, 850--866.

\item
Wahba, G. (1975). Smoothing noisy data with spline functions.
{\it Numerische Mathematik}, 24, 383--393.

\item
Wu, Y., and Liu, Y. (2009). Stepwise multiple quantile regression estimation
using non-crossing constraints. {\it Statistics and Its Interface}, 2, 299--310.

\item
Yoshida, T. (2021). Quantile function regression and variable selection for sparse models. 
{\it Canadian Journal of Statistics}, 49, 1196--1221.

\end{description}
}

\newpage

\section*{Appendix: R Functions}

The following functions are implemented in the R package `qfa' (version $\ge$ 5.0)
available at \url{https://cran.r-project.org} and \url{https://github.com/thl2019/QFA}. 

\begin{itemize}
\item {\tt sqr}: a function that computes SQR solutions on a set of quantile levels 
with user-supplied or automatically selected smoothing parameter {\tt spar}. Solver options for 
linear SQR are  ``{\tt fnb}'' (default) for the  interior-point method of Koenker et al.\  (1994)
and  ``{\tt sfn}'' for a sparse-matrix implementation of the interior-point method  (Koenker and Ng 2005).
Solver options for cubic SQR are ``{\tt piqp}'' (default) for the proximal interior-point method of Schwan et al.\  (2023)  and  ``{\tt osqp}'' for the first-order method of Stellato et al.\ (2020). 
\item {\tt boot.sqr}: a function that generates bootstrap samples of the SQR estimates based on the output of {\tt sqr}.
\item {\tt sqr.plot}: a function that plots SQR and QR estimates as functions of  quantile level.
\item {\tt sqr\underline{ }deriv.plot}:  a function that plots the derivatives of SQR estimates as functions of quantile level.
\item {\tt sqr1.fit}: a low-level function that computes the linear SQR solution (SQR1)  on a set of quantile levels 
with  user-supplied  smoothing parameter  {\tt spar}.
\item {\tt sqr3.fit}: a low-level function that computes the cubic SQR solution (SQR3) on a set of quantile levels with  user-supplied   smoothing parameter  {\tt spar}.
\item {\tt sqr.fit}: a low-level function that computes the SQR solution in Li and Megiddo (2026) 
\end{itemize}

\end{document}

%% file: myfonts.tex
\DeclareMathAlphabet{\mybm}{OT1}{ptm}{b}{it}

\usepackage{mathptmx}
\usepackage{euscript}
\usepackage{latexsym}
\usepackage{amsmath}
\usepackage{amsfonts}
\usepackage{amssymb}
\usepackage{ulem}
\usepackage{amsthm}

\newtheorem{thm}{Theorem}

\newtheorem{lem}{Lemma}

{\theoremstyle{remark} 
}

{\theoremstyle{definition} 

}

\newcommand{\BIBO}{{\small \BIBO}}

\newcommand{\lzero}{\rightarrow 0}

\newcommand{\eq}{\begin{eqnarray*}}
\newcommand{\eqq}{\end{eqnarray*}}
\newcommand{\eqn}{\begin{eqnarray}}
\newcommand{\eqqn}{\end{eqnarray}}

\newcommand{\eqb}{\begin{align*}}
\newcommand{\eqqb}{\end{align*}}

\newcommand{\E}{\mathrm{E}}

\newcommand{\diag}{\mathrm{diag}}
\newcommand{\tr}{\mathrm{tr}}

\newcommand{\eqna}{\begin{align}}
\newcommand{\eqqna}{\end{align}}

\newcommand{\vect}{\mathrm{vec}}

\newcommand{\bA}{\mathbf{A}}

\newcommand{\bD}{\mathbf{D}}
\newcommand{\bC}{\mathbf{C}}

\newcommand{\bI}{\mathbf{I}}

\newcommand{\bG}{\mathbf{G}}

\newcommand{\bP}{\mathbf{P}}
\newcommand{\bQ}{\mathbf{Q}}

\newcommand{\bX}{\mathbf{X}}
\newcommand{\ba}{\mathbf{a}}
\newcommand{\bb}{\mathbf{b}}
\newcommand{\bc}{\mathbf{c}}
\newcommand{\bd}{\mathbf{d}}

\newcommand{\bx}{\mathbf{x}}
\newcommand{\by}{\mathbf{y}}
\newcommand{\bz}{\mathbf{z}}

\newcommand{\bu}{\mathbf{u}}
\newcommand{\bv}{\mathbf{v}}
\newcommand{\br}{\mathbf{r}}
\newcommand{\bs}{\mathbf{s}}

\newcommand{\bw}{\mathbf{w}}

\newcommand{\bbR}{\mathbb{R}}

\newcommand{\0}{{\bf 0}}
\newcommand{\1}{{\bf 1}}

\newcommand{\cF}{\EuScript{F}}

\newcommand{\cT}{\EuScript{T}}

\newcommand{\ep}{\epsilon}

\newcommand{\sig}{\sigma}

\newcommand{\del}{\delta}

\newcommand{\bmdel}{\pmb{\delta}}
\newcommand{\bmbeta}{\pmb{\beta}}
\newcommand{\bmgam}{\pmb{\gamma}}

\newcommand{\bmphi}{\pmb{\phi}}

\newcommand{\bmth}{{\pmb{\theta}}}

\newcommand{\bmxi}{\pmb{\xi}}

\newcommand{\bmeta}{\pmb{\eta}}
\newcommand{\bmzeta}{\pmb{\zeta}}

\newcommand{\bmlam}{\pmb{\lambda}}

\newcommand{\bOm}{\pmb{\Omega}}

\newcommand{\bPhi}{\pmb{\Phi}}

